\documentclass[a4paper,UKenglish,cleveref, autoref, thm-restate]{lipics-v2021}

\usepackage{todonotes}
\usepackage{tikz,url}
\usepackage{tikz-cd}
\usetikzlibrary{arrows,automata,decorations}
\usepackage{amsmath,amssymb,amsthm}
\usepackage{epsfig}
\usepackage{mathrsfs}
\usepackage{mathpartir}
\usepackage{tikzit}

\tikzstyle{white vertex}=[fill=white, draw=black, shape=circle, thick, inner sep=1pt, minimum size=11pt]
\tikzstyle{green vertex}=[fill=green, draw=black, shape=circle, thick, inner sep=1pt, minimum size=11pt]
\tikzstyle{cyan vertex}=[fill=cyan, draw=black, shape=circle, thick, inner sep=1pt, minimum size=11pt]
\tikzstyle{magenta vertex}=[fill=magenta!60, draw=black, shape=circle, thick, inner sep=1pt, minimum size=11pt]
\tikzstyle{empty node}=[fill=white, shape=circle, inner sep=1pt]

\tikzstyle{filled path}=[-, fill=blue!60, thick, fill opacity=0.15]
\tikzstyle{thick edge}=[-, thick]

\usepackage{macros}

\title{Simplicial Models for the Epistemic Logic of Faulty Agents}

\author{\'Eric Goubault}{LIX, CNRS, \'Ecole Polytechnique, Institut Polytechnique de Paris, Paris, France}{eric.goubault@polytechnique.edu}{https://orcid.org/0000-0002-3198-1863}{}

\author{Roman Kniazev}{LIX, CNRS, \'Ecole Polytechnique, Institut Polytechnique de Paris, Paris, France \and Université Paris-Saclay, ENS Paris-Saclay, CNRS, LSV, 91190 Gif-sur-Yvette, France}{roman@kameronton.com}{https://orcid.org/0009-0006-7495-9793}{}

\author{J\'er\'emy Ledent}{Université Paris Cité, CNRS, IRIF, F-75013, Paris, France}{jeremy.ledent@irif.fr}{https://orcid.org/0000-0001-7375-4725}{}

\author{Sergio Rajsbaum}{Instituto de Matem\'aticas, UNAM, CDMX 04510, Mexico \and on leave at  LIX, \'Ecole Polytechnique and IRIF, Université Paris Cité }{rajsbaum@im.unam.mx}{https://orcid.org/0000-0002-0009-5287}{}

\authorrunning{\'E. Goubault, R. Kniazev, J. Ledent and S. Rajsbaum}

\Copyright{\'Eric Goubault, Roman Kniazev, J\'er\'emy Ledent and Sergio Rajsbaum}



\category{} 

\relatedversion{} 

\funding{Éric Goubault was partially funded by AID project CIEDS/FARO. Sergio Rajsbaum received additional support from ANR project DUCAT (ANR-20-CE48-0006), and Fondation Sciences Mathématiques de Paris (FSMP). }

\nolinenumbers 

\hideLIPIcs  

\EventEditors{}
\EventNoEds{0}
\EventLongTitle{}
\EventShortTitle{}
\EventAcronym{}
\EventYear{}
\EventDate{}
\EventLocation{}
\EventLogo{}
\SeriesVolume{}
\ArticleNo{}

\begin{document}
\maketitle

\begin{abstract}
In recent years, several authors have been investigating \emph{simplicial models}, a model of epistemic logic based on higher-dimensional structures called simplicial complexes.
In the original formulation of~\cite{gandalf-journal}, simplicial models are always assumed to be \emph{pure}, meaning that all worlds have the same dimension.
This is equivalent to the standard~$\Sfive$ semantics of epistemic logic, based on Kripke models.
By removing the assumption that models must be pure, we can go beyond the usual Kripke semantics and study epistemic logics where the number of agents participating in a world can vary.
This approach has been developed in a number of papers~\cite{Ditmarsch21, stacs22, goubaultSemisimplicialSetModels2023}, with applications in fault-tolerant distributed computing where processes may crash during the execution of a system.
A difficulty that arises is that subtle design choices in the definition of impure simplicial models can result in different axioms of the resulting logic.
In this paper, we classify those design choices systematically, and axiomatize the corresponding logics.
We illustrate them via distributed computing examples of synchronous systems where processes may crash.
%
\end{abstract}

\section{Introduction}

 


Logics for reasoning about multi-agent systems have been thoroughly studied, and are of interest to various 
research areas, including logic, artificial intelligence, economics, game theory~\cite{HWW2012survey}.
They are of particular interest to distributed systems since the early 1980's, showing the fundamental role of notions such as common knowledge~\cite{fagin,Moses2016}.
Modal epistemic logics are used, with a 
 language extending propositional logic by adding modalities $K_a$ representing the knowledge of each agent $a$.
 
The success of modal logics for reasoning about multi-agent systems is based on Kripke  semantics, built around the notion of ``possible world'' representing the \emph{state} of the system.
States and their relations are formally represented in Kripke models, where a binary relation for each agent $a$
is  taken to mean that $a$ cannot tell two states apart.
 This  classic \emph{possible worlds} relational structure  was developed by Rudolf Carnap, Stig Kanger, Jakko Hintikka and Saul Kripke in the late 1950's and early 1960's.

\subparagraph{From global states to local states.}
However, the intimate relationship between distributed computing and algebraic topology discovered in 1993~\cite{herlihyetal:2013} showed the importance of moving from using worlds as the primary object, to \emph{perspectives} about the worlds.
After all, what exists in a distributed system is only the local states of the agents and events observable within the system. 
The world, namely the global state of the system, consists of the set of local states of the agents, and in some cases the state of the environment, such as messages in transit or the state of the shared memory.
Thus,  a world  is an abstraction that may be useful to reason about the system, but not directly observable by the agents.

This point of view led to topological models of distributed systems, via a simplicial complex constructed using the local states as vertices and the global states as simplexes.
Remarkably, it was shown that there are topological invariants that are preserved while the agents communicate with each other, that in turn determine which distributed tasks can be solved, or how fast they can be solved. A fruitful theory has been developed since then (see~\cite{herlihyetal:2013} for an overview), for a variety of message passing and shared memory systems, where synchronous or asynchronous processes may fail.

The topological theory of distributed computability shows that  the  power of a distributed system to solve input/output tasks is determined by multi-dimensional indistinguishablity relations by sets of local states, rather than in the binary indistinguishability relations between pairs of global states defined in a Kripke structure.
The solvability of some tasks such as \emph{consensus} depends only on the one-dimensional (graph) connectivity of the Kripke structure of global states, and hence is intimately related to common knowledge.
However, other tasks are known whose solvability depends on the higher dimensional connectivity properties of the simplicial complex of local states.
Notable examples of such tasks are $\epsilon$-\emph{approximate agreement,} where process start with inputs in a Euclidean space of some dimension $d$, and communicate to decide on values at distance $\epsilon$ away from  each other, in the convex hull of their inputs~\cite{multidim2015}.
Another such example is $k$-\emph{set agreement}, where agents agree on at most $k$ different input values~\cite{HS99}. 

\subparagraph{From Kripke models to  simplicial models.}
The realization that distributed computability is of a topological nature motivated  the development of a formal semantics of epistemic logic formulas in terms of simplicial models~\cite{gandalf-journal}.
A new class of models was introduced, based on simplicial complexes, which is equivalent to the usual Kripke model semantics for $\Sfive$. 
 Tools were provided to reason about  solvability of distributed tasks such as consensus, approximate agreement and equality negation~\cite{gandalf-journal,DitmarschGLLR21}, as well as $k$-set agreement~\cite{yagiNishimura2020TR}. Bisimilarity of simplicial models was studied in~\cite{Ditmarsch2020KnowledgeAS}, and connections with covering spaces in~\cite{DitmarschGLLR21}.  
 
 Interestingly, the use of simplicial complexes exposes the importance of  the well-known notion of \emph{distributed knowledge}~\cite{HalpernM90}, to be a higher dimensional version of knowledge. With respect to a group of $k$ agents, distributed knowledge operates  
 by moving from  simplex to  simplex along shared faces of $k$ vertices corresponding to those agents. The use of distributed knowledge was crucial for the recent logical obstruction to the solvability of set agreement by Yagi and Nishimura~\cite{yagiNishimura2020TR}.

The categorical equivalence of~\cite{gandalf-journal} between $\Sfive$ Kripke models and simplicial models associates each world of the Kripke model with a facet of the corresponding simplicial model.
A core assumption of these models is that the same set of $n$ agents always participate in every possible world.
Because of this, every facet of the simplicial model is of the same dimension. 
Such models are called \emph{pure} simplicial models.
They can be used to analyse asynchronous models where crash failures are undetectable,
such as the basic \emph{wait-free} shared-memory model of computation~\cite{waitFree91}, where all interleavings of the individual operations of the agents are possible, to show that a
task is not wait-free solvable.

\subparagraph{When agents may die.}
In this paper we wish to  extend this equivalence to include simplicial models that are not pure.
The goal is to be able to reason about situations where not necessarily all agents are present in every world. A variety of such situations 
have been frequently studied in distributed computing, motivated by, to name just a few,  peer-to-peer systems with a
permanently evolving set of nodes~\cite{dynSys05}, in robot systems~\cite{LNCS11340}, in concurrent computing
where the set of processes can evolve~\cite{Aguilera04}, in natural systems~\cite{OhRR23}, and in blockchains~\cite{Herlihy19}.

Another way the set of agents can vary is in fault-tolerant distributed computing, when the agent represents a hardware or software component that has failed by crashing.
Synchronous distributed systems where processes may fail by crashing
have been thoroughly studied since early on in distributed computability, and have served to develop the theory of knowledge
since e.g.\ the  seminal work of Dwork and Moses~\cite{DworkM90crash}, where a complete characterization of the number of rounds required to reach simultaneous consensus was given, in terms of common knowledge.
%
%
For more recent additional references on algorithmic work see e.g.~\cite{CastanedaFPRRT23,tour17} and on knowledge based work see e.g~\cite{Unbeatable2014,silence2020,HalpernP17}.
Lower bounds on the number of rounds needed to solve set agreement are proved using the topological structure of the induced simplicial complexes e.g.~\cite{ChaudhuriHLT00,HERLIHY20001}. We will discuss later on the corresponding impure complexes, depicted in  Figure~\ref{fig-synchEvol}.

\subparagraph{Contributions.}
 We introduce in this paper an epistemic logic whose semantics is naturally given by impure simplicial models.
When some agents are missing from a simplex,  in epistemic logic terms, we will say that agents may \emph{die}.
Semantics based on Kripke models has been very successful to study  synchronous crash-failure models, but mainly for consensus~\cite{fagin}.
By moving from Kripke models to impure simplicial models, we open the door to study tasks beyond consensus when not all agents are always present, whose solvability depends on a higher dimensional structure; such as
 {$k$-set} agreement, renaming~\cite{CASTANEDA2011229}, multi-dimensional agreement~\cite{multidim2015}.


We start by discussing the distributed knowledge operator.
 Then, we introduce a generalized notion of simplicial models, where any simplex can be marked as a world, not necessarily a facet.
As we have seen in our previous work~\cite{gandalf-journal}, simplicial models correspond to Kripke models that are \emph{proper}. While this was not restrictive when we considered only the knowledge operator $K_a\,\phi$, it becomes important when we  include distributed knowledge, $D_B\,\phi$.
Indeed, even in the standard setting of the logic $\mathbf{S5_n}$, proper Kripke models obey the axiom $\phi \Rightarrow D_A\,\phi$, where $A$ is the set of all agents, while this property might fail in non-proper models.
In our setting where some agents may die, we introduce a similar axiom called~\textbf{P} (see \cref{sec:Axiomatization}) for that purpose. 

This new model comes 
with a full proof of completeness with respect to our epistemic logic. 
Compared to the other proofs of completeness for $D_B$ found in the literature~\cite{BaltagS20,FaginHV92}, we have two differences: our Kripke models are transitive and symmetric but not necessarily reflexive; and we have extra axioms that are specific to simplicial models. The general structure of the proof is however  similar.

Finally we present
a brief discussion about applications to fault-tolerant distributed computing.
To study the dynamics of how processes can communicate and crash during the execution of a distributed protocol, we use a slightly modified version of \emph{communication patterns} \cite{Castaneda22pattern}.
Communication patterns are an alternative to the \emph{action models} of Dynamic Epistemic Logic (DEL), which is better suited to study distributed computing dynamics.
In the original formulation of communication patterns~\cite{Castaneda22pattern}, the communication graphs are always assumed to be reflexive; by relaxing this assumption, we can accommodate the possibility of crashing agents. 
Finally, in \cref{sec:extended-example}, to exemplify how our logical framework can be leveraged to prove impossibility results in distributed computing, we study a classic example of a non-pure protocol complex in distributed computing: the synchronous crash failures model of computation~\cite{DworkM90crash}.
This model has been exploited in~\cite{ChaudhuriHLT00,HERLIHY20001} to establish a lower bound on the number of rounds required to solve set agreement.
Notice in \cref{fig-synchEvol} that the protocol complex is no longer a subdivision of the input complex, as in the asynchronous wait-free case.
Due to the possibility of crashes, holes and lower-dimensional simplexes appear after the first round.

\begin{figure}[h]
    \centering
    \includegraphics[scale=0.6]{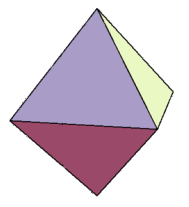}
    \qquad
    \includegraphics[scale=0.6]{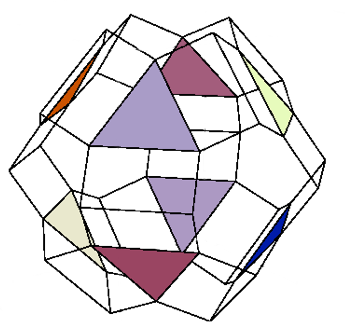}
    \qquad
    \includegraphics[scale=0.6]{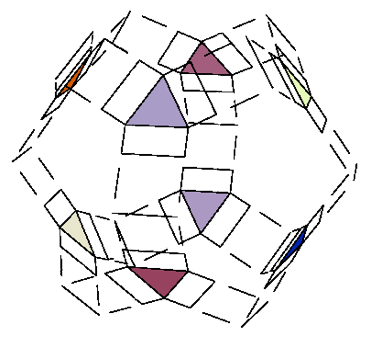}
    \caption{An input complex for three agents starting with binary inputs (left). Then the complex of local states after one round (middle), and two rounds (right). At most one agent may die~\cite{HERLIHY20001}.}
    \label{fig-synchEvol}
\end{figure}


\subparagraph{Relationship with previous papers.}
This article is an extended version of our conference paper~\cite{stacs22}, and also includes some ideas from a sequel conference paper~\cite{goubaultSemisimplicialSetModels2023}.
The definition of generalized simplicial models as formulated in \cref{def:generalized-simplicial-model} is new; it subsumes the impure simplicial models studied in~\cite{stacs22}, but it is strictly included in the so-called \emph{epistemic covering models} studied in~\cite{goubaultSemisimplicialSetModels2023}.
Contrary to the two conference papers, we include all proofs, including a fully detailed proof of completeness in \cref{sec:completeness}.
Moreover, in \cref{sec:dynamics,sec:application}, we go beyond the static setting studied until now, and introduce a new framework to study the dynamics of distributed communication with crashes.
\begin{itemize}
\item Compared to~\cite{stacs22}, we extended both the logic (adding the distributed knowledge operator), and the class of models that we consider (allowing worlds that are not facets of the simplicial complex).
Using distributed knowledge is crucial to study higher-dimensional geometric properties of models.
It also makes explicit the role of \emph{proper} models: the peculiar ``single-agent'' axiom \textbf{SA} of~\cite{stacs22} is now subsumed by our axiom of \emph{properness}~\textbf{P}.
This shows that there is nothing specific about the worlds with only one agent; we just lacked the distributed knowledge operator to express this in higher dimensions.
Allowing models there some worlds are not facets is important from the point of view of distributed computing, as it allows to model situations with \emph{undetectable crashes} (see \cref{sec:dynamics}).
\item Compared to~\cite{goubaultSemisimplicialSetModels2023}, the class of models that we study here is less general: we do not allow non-proper behavior, and we do not allow models with a semi-simplicial set geometric structure (a.k.a.\ pseudo-models, using the Kripke model terminology).
Both of those features are somewhat cumbersome to work with, and are rarely needed for distributed computing applications.
In particular, properly defining semi-simplicial sets involves some fairly advanced categorical lingo.
Here, we prefer to stay within the framework of simplicial complexes, and keep the paper easily accessible to readers unfamiliar with category theory.
\end{itemize}

\subparagraph{Related work.}

A line of work started by Dwork and Moses~\cite{DworkM90crash} studied in great detail the synchronous crash failures model from an epistemic logic perspective.
However, in their approach, the crashed processes are treated the same as the active ones, with a distinguished local state ``\texttt{fail}''.
In that sense, all agents are present in every state, hence they still model the usual epistemic logic $\Sfive$.
Instead of changing the underlying model as we do here, they introduce new knowledge and common knowledge operators that take into account the non-rigid set of agents (see e.g.~\cite{FHMVbook}, Chapter 6.4).

There are two other works that we are aware of, that considered the problem of defining a semantics of knowledge for possibly impure simplicial complexes.
Vel\'azquez-Cervantes~\cite{ThesisM2019} studies
projections  from impure complexes to pure sub-complexes, and algorithmic transformations between Kripke models and simplicial complexes.
More relevant to our purpose is the paper of van Ditmarsch~\cite{Ditmarsch21}, who describes a two-staged semantics with a \emph{definability relation} prescribing which formulas can be interpreted, on top of which the usual \emph{satisfaction relation} is defined.
This results in a three-valued logic, where formulas can be true, false or undefined.
A complete axiomatization of this logic was later established in~\cite{Ditmarsch22complete}, and it ends up being quite peculiar: for instance, it does not obey Axiom \textbf{K}, which is the common ground of all Kripke-style modal logics.
In contrast, we take a more systematic approach: we first establish a tight categorical correspondence between simplicial models and Kripke models.
Via this correspondence, we translate the standard Kripke-style semantics to simplicial models.
This leads us to a more standard two-valued logic, based on the well-understood modal logic $\KBfour$.

In another related paper~\cite{hypergraph}, we proposed a third approach to study the epistemic logic of faulty agents.
In that work, we study a refinement of epistemic logic where formulas are separated into several sorts: ``agent formulas'' and ``world formulas''.
This avoids entirely the question of how to define the knowledge of a dead agent, since such a formula would be ill-typed.
This approach might constitute a bridge between the three-valued semantics of van Ditmarsch et al., and the the two-valued semantics presented here.
The results of~\cite{hypergraph} are formulated using so-called \emph{hypergraph models} rather than simplicial models. As we will see in \cref{rem:hypergraph-models}, this is not a fundamental difference, but simply a shift in perspective.

The example of synchronous crash failures that we study in \cref{sec:application} has also been considered in \cite{nakai2023partial}, concurrently with our paper.
However, some slight differences can be noted.
To formalize the dynamics, they introduce a variant of the DEL action models in which processes can crash; whereas we rely on a variant of communication pattern models (\cref{sec:dynamics}).
As expected, the resulting simplicial model for synchronous crash failures is the same.
Moreover, the obstruction formula used to prove impossibility is different: in~\cite{nakai2023partial}, the formula is specifically tailored to prove impossibility in one round, using three nested knowledge operators.
In contrast, we use a more general formula relying on the common knowledge operator.
Lastly, we discuss some other variants of consensus task specification in the presence of crashes.
Our main focus though, is to showcase how the epistemic logic machinery developed in this paper can be used to study concrete distributed computing problems.

\subparagraph{Plan of the paper.}
In \cref{sec:backsimp}, we briefly recall the equivalence between pure simplicial complexes and epistemic Kripke models, as originally studied in~\cite{gandalf-journal}.
In \cref{sec:semantics-DK}, we introduce \emph{generalized simplicial models} as a semantics for distributed knowledge.
We then define in \cref{sec:kripke} an equivalent class of Kripke models, called \emph{partial epistemic models}, and describe the formal relationship with simplicial models.
The main technical result of the paper is the completeness result, proved in full details in \cref{sec:completeness}.
Then, in \cref{sec:dynamics}, we define an update operator to study the dynamics of simplicial models, based on communication patterns.
And finally in \cref{sec:application}, as a proof of concept, we study the solvability of consensus in the synchronous message-passing model.

\section{Background on simplicial complexes and Kripke structures}

\label{sec:backsimp}

\subparagraph{Chromatic simplicial complexes.}
Simplicial complexes with vertices labeled with agent names have been used extensively in the field of fault-tolerant distributed protocols \cite{herlihyetal:2013}. They are defined as follows: 

\begin{definition}
\label{def:simplicial-complex}
A \emph{simplicial complex} is a pair $\C = \langle V,S \rangle$ where $V$ is a set, and $S \subseteq \Pow{V}$ is a family of non-empty subsets of $V$ such that:
\begin{itemize}
\item for all $v \in V$, $\{v\} \in S$, and
\item $S$ is downward-closed: for all $X \in S$, $Y \subseteq X$ implies $Y \in S$. 
\end{itemize}
Given a finite set~$A$ of colours, a \emph{chromatic simplicial complex} coloured by~$A$ is a triple $\langle V,S,\chi \rangle$ where~$\langle V,S \rangle$ is a simplicial complex, and~$\chi : V \to A$ is required to assign distinct colours to the elements of every $X \in S$.
\end{definition}

\label{def:puresimp}
Elements of $V$ (identified with singletons) are called \emph{vertices}.
Elements of $S$ are \emph{simplexes}, and the ones that are maximal w.r.t.\ inclusion are \emph{facets}.
The set of facets of~$\C$ is denoted~$\Facets(\C)$.
The \emph{dimension} of a simplex $X \in S$ is $\dim(X) = |X|-1$.
 A \emph{face} of a simplex $X$ is a subset $X'\subset X$.
A simplicial complex $C$ is \emph{pure} if all facets are of the same dimension. 

The condition of having distinct colours for vertices of a simplex $X$ implies that
given a set of colours $U$ of $\chi(X)$, there is a unique face of  $X$ colored with $U$.

Chromatic simplicial complexes 
can be arranged into a category, 
whose morphisms preserve simplex dimension:
\begin{definition}
\label{def:puresimpmor}
A chromatic simplicial map from $\C = \langle V,S,\chi \rangle$ to $\D = \langle V',S',\chi' \rangle$ is a function $f : V \to V'$ such that:
\begin{itemize}
\item $f$ maps simplexes to simplexes, i.e., for every $X \in S$, $f(X) \in S'$, and
\item $f$ respects colours, i.e., for every $v\in V$, $\chi'(f(v)) = \chi(v)$.
\end{itemize}
\end{definition}

We denote by $\SimCpx{A}$ the category of chromatic simplicial complexes coloured by~$A$, and 
$\PureSimCpx{A}$ the full sub-category of pure chromatic simplicial complexes on $A$. 

\subparagraph{Equivalence with epistemic frames.}
The traditional possible worlds semantics of (multi-agent) modal logics relies on the notion of Kripke frame.
In the following definition, we fix a finite set $A$ of \emph{agents}.

\begin{definition}
\label{def:kripkeframe}
\label{def:kripkeframemor}
A \emph{Kripke frame} $M = \la W, R \ra$ is given by a set of \emph{worlds}~$W$, together with an $A$-indexed family of relations on~$W$, $R : A \to \Pow{W \times W}$.
We write $R_a$ rather than $R(a)$, and $u\,R_a\,v$ instead of $(u,v) \in R_a$.
The relation $R_a$ is called the \emph{$a$-accessibility relation}.
%
%
Given two Kripke frames $M=\la W, R \ra$ and $N=\la W',R' \ra$, a \emph{morphism} from  $M$ to $N$ is a function $f : W \to W'$ such that for all $u, v \in W$, for all $a \in A$, 
$u\,{R_a}\,v$ implies $f(u)\,{R'_a}\,f(v)$.
\end{definition}

To model multi-agent epistemic logic $\Sfive$, we additionally require each relation~$R_a$ to be an equivalence relation.
When this is the case, we usually denote the relation by $\sim_a$, and call it the \emph{indistinguishability relation}.
For the equivalence class of~$w$ with respect to $\sim_a$, we write $[w]_{a} \subseteq W$.
Kripke frames satisfying this condition are called \emph{epistemic frames}.
An epistemic frame is \emph{proper} when two distinct worlds can always be distinguished by at least one agent: for all $w,w' \in W$, if $w \neq w'$ then $w \not \sim_a w'$ for some~$a \in A$.
In \cite{gandalf-journal}, we exploited an equivalence of categories between pure chromatic simplicial complexes and proper Kripke frames, to give an interpretation of $\Sfive$ on simplicial models. This allowed us to apply epistemic logics to study distributed tasks.

\begin{theorem}[see \cite{gandalf-journal}]
\label{thm:gandalf}
The category of pure chromatic simplicial complexes $\PureSimCpx{A}$ is equivalent to the category of proper epistemic frames $\ProperEFrame{A}$.
\end{theorem}

\begin{example}\label{ex:basicDuality}
The picture below shows an epistemic frame (left) and its associated chromatic simplicial complex (right).
The three agents $a, b, c$, are represented as
colours blue, magenta and green (respectively) on the vertices of the simplicial complex.
The three worlds $\{w_1, w_2, w_3\}$ of the epistemic frame correspond to the three facets (triangles) of the simplicial complex.
The $c$-labeled edge between the two worlds $w_2$ and $w_3$ indicates that $w_2 \sim_c w_3$.
Correspondingly, the two facets $w_2$ and $w_3$ of the simplicial complex share a common vertex, coloured in green (agent~$c$).
Similarly, the two facets $w_1$ and $w_2$ share their $ab$-coloured edge.

\begin{center}
\begin{tikzpicture}[auto,cloudgrey/.style={draw=black,thick,circle,fill=cyan,inner sep=1pt,minimum size=11pt},cloud/.style={draw=black,thick,circle,fill=magenta!60,inner sep=1pt,minimum size=11pt}, cloudblack/.style={draw=black,thick,circle,fill=green,inner sep=1pt,minimum size=11pt}]

\node (p) at (-1.5,0) {$w_1$};
\node (q) at (0,0) {$w_2$};
\node (r) at (1.5,0) {$w_3$};
\path (p) edge[bend left] node[above] {$a$} (q)
      (p) edge[bend right] node[below] {$b$} (q)
      (q) edge node[above] {$c$} (r);
 
\node at (3.25,0) {\Large $\cong$};

\draw[thick, draw=black, fill=lipicsLightGray, fill opacity=0.7]
  (5,0) -- (6,-0.577) -- (6,0.577) -- cycle;
\draw[thick, draw=black, fill=lipicsLightGray, fill opacity=0.7]
  (6,-0.577) -- (6,0.577) -- (7,0) -- cycle;
\draw[thick, draw=black, fill=lipicsLightGray, fill opacity=0.7]
  (7,0) -- (8,-0.577) -- (8,0.577) -- cycle;
\node (p') at (5.65,0) {$w_1$};
\node (q') at (6.35,0) {$w_2$};
\node (r') at (7.65,0) {$w_3$};
\node[cloudblack] (b1) at (5,0) {$c$};
\node[cloudgrey] (g1) at (6,-0.577) {$a$};
\node[cloud] (w1) at (6,0.577) {$b$};
\node[cloudblack] (b2) at (7,0) {$c$};
\node[cloudgrey] (g2) at (8,-0.577) {$a$};
\node[cloud] (w2) at (8,0.577) {$b$};
\end{tikzpicture}
\end{center}
\end{example}

\section{Simplicial semantics of Epistemic logic with Distributed Knowledge}

\label{sec:semantics-DK}

Let $\AP$ be a countable set of atomic propositions and $A$ a finite set of agents.
We consider the language $\mathcal{L}_D$ of epistemic logic with the \emph{distributed knowledge} operator~\cite{fagin,halpernmoses:1990}, generated by the following BNF
grammar:
$$
\varphi ::= p \mid \neg\varphi \mid \varphi \land \varphi \mid
D_B\,\varphi \qquad p \in \AP,\ B \subseteq A,\ B \neq \emptyset
$$
Other standard operators can be derived from the basic ones as follows:
\begin{mathpar}
\phi \lor \psi := \neg(\neg \phi \land \neg \psi) \and
\phi \Rightarrow \psi := \neg \phi \lor \psi \and
\true := p \lor \neg p \and
\false := \neg \true \\
K_a\,\phi := D_{\{a\}}\,\phi \and
E_B\,\phi := \bigwedge_{a \in B} K_a\,\phi
\end{mathpar}

The distributed knowledge operator $D_B\,\phi$ models, intuitively, what a group~$B$ of agents would know if they were able to combine their individual knowledge (for example, via perfectly reliable communication).
Another way to explain it is that we view the group~$B$ of agents as a single entity, which is able to distinguish to possible worlds whenever at least one agent $a \in B$ can distinguish them.
Thus, in the usual Kripke-style semantics for epistemic logic, the indistinguishability relation $\sim_B$ of the group~$B$ is obtained as the intersection of the relations of all the agents in~$B$:\quad $\sim_B \;=\; \bigcap_{a \in B} \sim_a$.

Distributed knowledge should not be confused with another group knowledge operator, the \emph{everybody knows} operator $E_B\,\phi$, which asserts that every agent in the group~$B$ knows the formula~$\phi$.
Technically, this amounts to taking the union of the indistinguishability relations of the agents $a \in B$, rather than the intersection.
Another distinction between the two operators is that, given some agent $a \in B$, we have $E_B\,\phi \Rightarrow K_a\,\phi$ but $K_a\,\phi \Rightarrow D_B\,\phi$.

In the next section, we define the semantics of distributed knowledge for simplicial models.
As we will see, this operator is crucial for our topological approach since it makes use of the higher-dimensional connectivity between adjacent simplexes.
Indeed, while the operator $K_a\,\phi$ only looks at whether two simplexes share a common vertex, the operator $D_B\,\phi$ is concerned with whether two simplexes share a common face of higher dimension (edge, triangle, etc).
The distributed knowledge operator is also crucial for applications to distributed computing such as the \emph{$k$-set agreement tasks}~\cite{hoshino22,yagiNishimura2020TR}.

\subsection{Generalized simplicial models}

\label{sec:gensimpmod}
Since the introduction of simplicial models in~\cite{gandalf-journal}, several variants of this notion have been studied. Indeed, there is a number of design choices that can be made:
\begin{itemize}
\item The underlying topological structure of the model. In the original paper~\cite{gandalf-journal}, the model is assumed to be a \emph{pure} simplicial complex. This yields a notion of model which is equivalent to standard Kripke models, but is quite restrictive from a topological point of view.
Subsequent works have lifted this condition: both~\cite{Ditmarsch21} and \cite{stacs22} (the conference version of this paper) consider possibly impure simplicial complexes.
An even the larger class of models considered in a sequel of this work~\cite{goubaultSemisimplicialSetModels2023},  obtained by considering \emph{pre-simplicial sets}, a strict generalization of simplicial complexes. 
\item Atomic propositions on the worlds vs.\ vertices. In the epistemic logic literature, the notion of Kripke model usually contains a valuation function, which equips each world with a set of atomic propositions.
This contrasts with the usual practice in distributed computing, which labels the vertices of a model with atomic propositions instead.
Previous papers on simplicial models have taken the distributed computing approach. In~\cite{gandalf-journal}, we showed that this choice results in an extra axiom, dubbed the ``Axiom of locality'', which asserts that every atomic proposition belongs to a particular agent, who must always know whether this proposition is true or false.
Here, as we did in~\cite{stacs22}, we label directly the worlds of a model, in order to avoid dealing with this locality condition.
This is strictly more general: local models as defined in~\cite{gandalf-journal} are a strict subclass of the models presented here.
In \cref{sec:dynamics}, we will restrict to local models for distributed computing applications.
\item Worlds are facets vs.\ simplexes. In the original paper on simplicial models~\cite{gandalf-journal}, epistemic formulas could only be interpreted in a \emph{facet} of a simplicial model; hence, we used the words ``world'' and ``facet'' interchangeably.
The idea that any simplex (not necessarily a maximal one) might be a world was initially raised in~\cite{Ditmarsch2020KnowledgeAS}, and further explored in~\cite{Ditmarsch21}. In this approach, any simplex, without restriction, is considered to be a world, and we can interpret epistemic formulas on it.
Here, we take an even more general stance, and require the model to specify a set of worlds, which may contain only the facets, or all simplexes, or any set of simplexes in-between the two. As we will see, there are distributed computing applications where the set of worlds is indeed something ``in-between''.
\end{itemize}

We now introduce a notion of model based on (possibly not pure) simplicial complexes.
They are equipped with a distinguished subset of simplexes called the \emph{worlds}, which contains all the facets, and a valuation function that assigns to each world the set of all atomic propositions that are true in this world.

\begin{definition} \label{def:generalized-simplicial-model}
A \emph{(generalized) simplicial model} $\C = \la V, S, \chi, W, \ell \ra$ over the set of agents~$A$ consists of a chromatic simplicial complex $\la V,S,\chi \ra$ together with a distinguished set of worlds~$W$ such that $\Facets(\C) \subseteq W \subseteq S$, and a labelling $\ell : W \to \Pow{\AP}$
that associates with each world $w \in W$ a set of atomic propositions.
\end{definition}

\begin{remark} \label{rem:generalized-models}
Let us explain how the class of simplicial models of \cref{def:generalized-simplicial-model} relates to those of previous papers.
In the conference version of this work~\cite{stacs22}, the models that we studied were exactly those such that $W = \Facets(\C)$ (here, we call them the ``minimal'' models).
In the original paper on simplicial models~\cite{gandalf-journal}, the class of models considered was even smaller: on top of being minimal, we further impose that $\C$ must be pure, and that the labelling $\ell$ of a facet must be given by the union of the local labellings on its vertices (we will properly define the class of ``local'' simplicial models in \cref{sec:dynamics}).
The models studied by van Ditmarsch in~\cite{Ditmarsch21} implicitly use the set of worlds $W = S$ instead (here, we call them the ``maximal'' models); however, we do not claim to cover this work since the satisfaction relation that we define on our models is very different from the one of~\cite{Ditmarsch21}.
Finally, the class of models considered in~\cite{goubaultSemisimplicialSetModels2023} is even larger than the one that we defined. It allows the underlying geometric structure to be a semi-simplicial set, rather than a simplicial complex. Moreover, it allows to have several copies of the same world (which we call being ``non-proper'' here).
Using the terminology of~\cite{goubaultSemisimplicialSetModels2023}, the simplicial models of \cref{def:generalized-simplicial-model} are exactly the epistemic covering models that are proper, have no empty worlds, and have standard group knowledge.
\end{remark}

A \emph{pointed simplicial model} $(\C,w)$ consists of a simplicial model $\C$ together with a distinguished world $w \in W$.
Given a pointed simplicial model $(\C,w)$, we define the \emph{satisfaction relation} $\C,w \models \phi$ by induction on the formula $\phi$, as follows.
\[ \begin{array}{lcl}
\C,w \models p & \text{iff} & p \in \ell(w) \\
\C,w \models \neg\phi & \text{iff} & \C,w \not\models \phi \\
\C,w \models \phi\et\psi & \text{iff} & \C,w \models \phi \text{ and } \C,w \models \psi \\
\C,w \models D_B\,\phi & \text{iff} & \C,w' \models \phi \text{ for all } w' \in W \text{ such that } B \subseteq \chi(w \inter w')
\end{array} \]

When the relation $\C,w \models \phi$ holds, we say that the formula $\phi$ is \emph{true} in the world $w$.
The first three clauses are the standard interpretation of propositional logic.
The one for distributed knowledge says the following: $D_B\,\phi$ is true in world $w$ when $\phi$ is true in every world $w'$ that shares a $B$-colored face with $w$.


We will study the following two important subclasses of simplicial models.

\begin{definition}
A simplicial model $\C = \la V, S, \chi, W, \ell \ra$ is called:
\begin{itemize}
\item  \emph{minimal}, when the set of worlds is the set of facets, i.e., $W = \Facets(\C)$.
\item \emph{maximal}, when the set of worlds is the set of all simplexes, i.e., $W = S$.
\end{itemize}
\end{definition}

\begin{example}
\label{ex:toy-examples}
Three toy examples are depicted below to illustrate some specific features of our models.
The three models (called $\C_1$, $\C_2$ and $\C_3$ from left to right) use the same set of agents, $A = \{a,b,c\}$. The agents $a$,$b$ and $c$ are depicted using colours blue, magenta and green, respectively.
In all three models, there are four worlds $W = \{w_1, w_2, w_3, w_4\}$.
We consider a unique atomic proposition called $p$ which is true exactly in the worlds $w_1$ and $w_2$. Thus, we have $\AP = \{p\}$, and the valuation function is given by $\ell(w_1) = \ell(w_2) = \{p\}$ and $\ell(w_3) = \ell(w_4) = \emptyset$.
Note that the models $\C_1$ and $\C_2$ are both minimal, since all the worlds are facets. On the other hand, model $\C_3$ is neither maximal nor minimal.
\begin{center}
	\begin{tikzpicture}[auto,scale=1.4,cloudgrey/.style={draw=black,thick,circle,fill=cyan,inner sep=1pt,minimum size=11pt},cloud/.style={draw=black,thick,circle,fill=magenta!60,inner sep=1pt,minimum size=11pt}, cloudblack/.style={draw=black,thick,circle,fill=green,inner sep=1pt,minimum size=11pt}]
	\draw[thick, draw=black, fill=blue!60, fill opacity=0.15]
	  (5,0) -- (6,0) -- (5.5,0.866) -- cycle;
	\draw[thick, draw=black, fill=blue!60, fill opacity=0.15] (6,0) -- (7,0) -- (6.5,0.866) -- cycle;
	\draw[thick, draw=black, fill=blue!60, fill opacity=0.15] (4.5,0.866) -- (5.5,0.866) -- (5,0) -- cycle;
	\draw[thick, draw=black, fill=blue!60, fill opacity=0.15] (5.5,0.866) -- (6.5,0.866) -- (6,0) -- cycle;
	\node[cloudgrey] (g1) at (5,0) {$a$};
	\node[cloudblack] (w1) at (6,0) {$c$};
	\node[cloud] (b2) at (5.5,0.866) {$b$};
	\node[cloudgrey] (w2) at (6.5,0.866) {$a$};
	\node[cloud] (g2) at (7,0) {$b$};
	\node[cloudblack] (b1) at (4.5,0.866) {$c$};
	\node at (5,0.55) {$w_4$};
	\node at (5.5,0.3) {$w_1$};
	\node at (6,0.55) {$w_2$};
	\node at (6.5,0.3) {$w_3$};
	\draw[thick] (w1) edge node[text opacity=0] {$w_0$} (g1); 
	\end{tikzpicture}
	\hspace{1.2cm}
	\begin{tikzpicture}[auto,scale=1.2,cloudgrey/.style={draw=black,thick,circle,fill=cyan,inner sep=1pt,minimum size=11pt},cloud/.style={draw=black,thick,circle,fill=magenta!60,inner sep=1pt,minimum size=11pt}, cloudblack/.style={draw=black,thick,circle,fill=green,inner sep=1pt,minimum size=11pt}]
	\draw[thick, draw=black, fill=blue!60, fill opacity=0.15]
	  (4,0) -- (5,-0.577) -- (5,0.577) -- cycle;
	\node[cloudgrey] (b1) at (4,0) {$a$};
	\node[cloudblack] (g1) at (5,-0.577) {$c$};
	\node[cloud] (w1) at (5,0.577) {$b$};
	\node[cloudblack] (g2) at (6.154,0.577) {$c$};
	\node[cloud] (w2) at (6.154,-0.577) {$b$};
	\node at (4.65,0) {$w_1$};
	\draw[thick] (w1) edge node[above] {$w_2$} (g2);
	\draw[thick] (g2) edge node[right] {$w_4$} (w2);
	\draw[thick] (w2) edge node[below] {$w_3$} (g1);
	\end{tikzpicture}
	\hspace{1.2cm}
	\begin{tikzpicture}[auto,rotate=30,scale=1.2,cloudgrey/.style={draw=black,thick,circle,fill=cyan,inner sep=1pt,minimum size=11pt},cloud/.style={draw=black,thick,circle,fill=magenta!60,inner sep=1pt,minimum size=11pt}, cloudblack/.style={draw=black,thick,circle,fill=green,inner sep=1pt,minimum size=11pt}]
	\draw[thick, draw=black, fill=blue!60, fill opacity=0.15]
	  (4,0) -- (5,-0.577) -- (5,0.577) -- cycle;
	\node at (4.65,0) {$w_1$};
	\node[cloudgrey] (b1) at (4,0) {$a$};
	\node[cloudblack] (g1) at (5,-0.577) {$c$};
	\node[cloud,label=above:{$w_4$}] (w1) at (5,0.577) {$b$};
	\draw[thick] (w1) edge node {$w_2$} (g1);
	\draw[thick] (b1) edge node {$w_3$} (w1);
	\draw[thick] (g1) edge node[text opacity=0] {$w_0$} (b1); 
	\end{tikzpicture}
	\[
	\begin{array}{@{\hskip 1.2cm}l@{\hskip 2cm}l@{\hskip 1.3cm}l}
	\C_1,w_1 \models \neg K_b\,p &
	\C_2,w_1 \models K_a\,p \land K_b\,p &
	\C_3,w_1 \models D_{\{b,c\}}\,p \\
	\C_1,w_1 \models \neg K_c\,p &
	\C_2,w_1 \models \neg K_c\,p &
	\C_3,w_1 \models \neg D_{\{a,b\}}\,p \\
	\C_1,w_1 \models D_{\{b,c\}}\,p &
	\C_2,w_4 \models K_a\,p &
	\C_3,w_2 \models D_{\{a,b\}}\,\false
	\end{array}
	\]
\end{center}
Let us comment some of the example formulas given above.
\begin{itemize}
\item Model $\C_1$ illustrates the topological meaning of distributed knowledge. In the world~$w_1$, agent~$b$ does not know~$p$, because the world~$w_4$ is indistinguishable (i.e., $w_4$ shares a $b$-coloured vertex with $w_1$).
Similarly, agent~$c$ does not know~$p$, because of world~$w_3$.
However, the group $\{b,c\}$ has distributed knowledge of~$p$.
Indeed, to check that $D_{\{b,c\}}\,p$ holds in world~$w_1$, we have to check all the worlds that share a $bc$-coloured edge with~$w_1$. The only worlds which qualify are $w_1$ and $w_2$, and in both cases, $p$ is true.
\item Model $\C_2$ is an example of a model where the simplicial complex is not pure: it has a facet of dimension 2 ($w_1$) and three facets of dimension 1 ($w_2, w_3, w_4$).
In the worlds $w_2$, $w_3$ and $w_4$, only the agents $b$ and $c$ are alive: the agent $a$ is not participating. However, $b$ and $c$ may or may not be aware of whether $a$ is alive or dead.
Interestingly, we can still evaluate formulas talking about dead agents: in world $w_4$, we have $\C_2,w_4 \models K_a\,p$. Indeed, there is no world that shares an $a$-coloured vertex with $w_4$ (since~$w_4$ has no $a$-coloured vertex to begin with!), so the condition is vacuously true.
In fact, we could even write: $\C_2,w_4 \models K_a\,\false$.
\item Model $\C_3$ is an example of a model which has \emph{sub-worlds}. This situation arises when some agents may die, and when none of the remaining agents is aware of it. In the picture, all three agents are alive in world~$w_1$; $a$ is dead in~$w_2$; $c$ is dead in~$w_3$; and both $a$ and $c$ are dead in~$w_4$.
One can check, for example, that in world~$w_1$ the formula $D_{\{a,b\}}\,p$ is not satisfied, because~$w_1$ shares and $ab$-coloured edge with the world $w_3$ where~$p$ does not hold.
As in the model $\C_2$, some formulas can be vacuously true when they involve the knowledge of dead agents: for example, $D_{\{a,b\}}\,\false$ holds in world~$w_2$.
\end{itemize}
\end{example}

In \cref{ex:toy-examples}, we introduced some informal vocabulary such as ``alive'' agents or ``sub-worlds''. We now define these notions formally.

\begin{definition}[alive, dead] \label{def:alive}
Let $\C = \la V, S, \chi, W, \ell \ra$ be a simplicial model, $w \in W$ a world of $\C$, and $a \in A$ an agent.
We say that \emph{$a$ is alive in $w$} when $a \in \chi(w)$.
Similarly, agent $a$ is \emph{dead} in $w$ when $a \not \in \chi(w)$.
\end{definition}

\begin{definition}[sub-world] Given a simplicial complex $\la V, S \ra$, and two simplexes $X, Y \in S$, we say that \emph{$X$ is a sub-simplex\footnote{In topology, $X$ is often called a \emph{face} of $Y$, but we prefer to avoid the confusion with the word ``facet''.} of $Y$} when $X \subseteq Y$.
Similarly in a simplicial model $\la V, S, \chi, W, \ell \ra$,
we say that a world $w_1 \in W$ is a \emph{sub-world} of $w_2 \in W$ when $w_1 \subseteq w_2$.
\end{definition}

\begin{example}[Synchronous broadcast with one crash]
\label{ex:message-passing}
The picture below shows a simplicial model after one round of the synchronous broadcast protocol with one crash failure, for three processes $a$, $b$ and $c$.
This distributed computing model will be studied in full detail in \cref{sec:extended-example}.
This model is comprised of 10 facets $w_0, \ldots, w_9$ of various dimension. World $w_0$ (of dimension 2) corresponds to an execution where no crash occurred, so all three agents are alive. On the other hand, in worlds $w_1$, $w_2$ and $w_3$ (of dimension 1), agent~$c$ has crashed so only the agents~$a$ and~$b$ are alive.
If we consider the minimal model where the set of worlds is exactly the facets $W = \{w_0, \ldots, w_9\}$, then we are modelling a protocol with \emph{detectable crashes}. That is, we would be assuming implicitly that whenever a process crashes, one of the remaining processes has to be aware of it.
If, instead, we want to consider a model where crashes might not always be detectable, we should also include some sub-worlds of this model.
Note that we do not attach atomic propositions to the worlds here since this will be done in \cref{sec:extended-example} where we describe this distributed computing example in more detail.
\begin{center}
	\begin{tikzpicture}[scale=1.9, auto, font={\small}, baseline={(0,0)}]
	\begin{pgfonlayer}{nodelayer}
		\node [style=cyan vertex] (0) at (0, 1) {$a$};
		\node [style=green vertex] (1) at (1, -0.5) {$c$};
		\node [style=magenta vertex] (2) at (-1, -0.5) {$b$};
		\node [style=green vertex] (3) at (1.5, 2) {$c$};
		\node [style=cyan vertex] (4) at (2.5, 0.5) {$a$};
		\node [style=magenta vertex] (5) at (1, -2.25) {$b$};
		\node [style=green vertex] (6) at (-1, -2.25) {$c$};
		\node [style=magenta vertex] (7) at (-1.5, 2) {$b$};
		\node [style=cyan vertex] (8) at (-2.5, 0.5) {$a$};
		\node [style=none] (12) at (-1, -0.5) {};
		\node [style=none] (13) at (0, 1) {};
		\node [style=none] (14) at (1, -0.5) {};
		\node [style=none] (15) at (0, 0) {$w_0$};
	\end{pgfonlayer}
	\begin{pgfonlayer}{edgelayer}
		\draw [style=thick] (2) to node [inner sep=2pt] {$w_3$} (8);
		\draw [style=thick] (8) to node [inner sep=2pt] {$w_2$} (7);
		\draw [style=thick] (7) to node [inner sep=1pt] {$w_1$} (0);
		\draw [style=thick] (0) to node [inner sep=1pt] {$w_9$} (3);
		\draw [style=thick] (3) to node [inner sep=2pt] {$w_8$} (4);
		\draw [style=thick] (4) to node [inner sep=2pt] {$w_7$} (1);
		\draw [style=thick] (1) to node [inner sep=2pt] {$w_6$} (5);
		\draw [style=thick] (5) to node [inner sep=2pt] {$w_5$} (6);
		\draw [style=thick] (6) to node [inner sep=2pt] {$w_4$} (2);
		\draw [style=filled path] (13.center)
			 to (14.center)
			 to (12.center)
			 to cycle;
	\end{pgfonlayer}
\end{tikzpicture}

\end{center}
\end{example}

\begin{example}[Immediate snapshot model with initial crashes] \label{ex:immediate-snapshot}
An example of a distributed computing model which is neither minimal nor maximal is the immediate snapshot model with \emph{initial crashes} (see e.g.\ \cite[Chapter 8]{herlihyetal:2013}).
This means that a process can only crash before the start of the computation. In other words, the set of participating processes is not known in advance.
Thus, in the picture below, there are 3 vertices $w_0, w_1, w_2$ corresponding to solo executions where only one process is alive; 9 edges $w_3, \ldots, w_{11}$ corresponding to executions where only two processes are participating; and
13 two-dimensional worlds $w_{12}, \ldots, w_{24}$ where all three processes are participating.
  \begin{center}
  \begin{tikzpicture}[scale=0.9, font=\footnotesize,
  cloud/.style={draw=black,thick,circle,fill=cyan,inner sep=1.5pt,minimum size=6pt}, cloudblack/.style={draw=black,thick,circle,fill=magenta!60,inner sep=1.5pt,minimum size=6pt},  cloudgrey/.style={draw=black,thick,circle,fill=green,inner sep=1.5pt,minimum size=6pt}]
  \filldraw[fill=blue!60, fill opacity=0.15] (0,0) -- (6,0) -- (3,5.2) -- cycle;
  \node[cloudblack, label=left:{$w_1$}] (p) at (0,0) {$b$};
  \node[cloudgrey, label=right:{$w_2$}] (q) at (6,0) {$c$};
  \node[cloud, label=above:{$w_0$}] (r) at (3,5.2) {$a$};

  \draw (p) -- (q) node[cloudblack](p-pq) [pos=0.66] {$b$}
                   node[cloudgrey](q-pq) [pos=0.33] {$c$}; 
  \draw (p) -- (r) node[cloudblack](p-pr) [pos=0.66] {$b$}
                   node[cloud](r-pr) [pos=0.33] {$a$}; 
  \draw (q) -- (r) node[cloudgrey](q-qr) [pos=0.66] {$c$}
                   node[cloud](r-qr) [pos=0.33] {$a$}; 
  \node[cloudblack] (p-pqr) at (4.5-5.2/6*1,2.6-1/2*1) {$b$};
  \node[cloudgrey] (q-pqr) at (1.5+5.2/6*1,2.6-1/2*1) {$c$};
  \node[cloud] (r-pqr) at (3,1) {$a$};
  \draw (p) -- (q-pqr) -- (r-pqr) -- (p-pqr) -- (q) -- (r-pqr) -- (p);
  \draw (r) -- (q-pqr) -- (p-pqr) -- (r);
  \draw (p-pq) -- (r-pqr) -- (q-pq);
  \draw (p-pr) -- (q-pqr) -- (r-pr);
  \draw (q-qr) -- (p-pqr) -- (r-qr);
  \draw (p) edge node[auto] {$w_{5}$} (r-pr); 
  \draw (r-pr) edge node[auto] {$w_{4}$} (p-pr); 
  \draw (p-pr) edge node[auto] {$w_{3}$} (r); 
  \draw (r) edge node[auto] {$w_{11}$} (q-qr); 
  \draw (q-qr) edge node[auto] {$w_{10}$} (r-qr); 
  \draw (r-qr) edge node[auto] {$w_{9}$} (q);
  \draw (q) edge node[auto] {$w_{8}$} (p-pq);
  \draw (p-pq) edge node[auto] {$w_{7}$} (q-pq); 
  \draw (q-pq) edge node[auto] {$w_{6}$} (p);
  \node at (2.4,3.6) {$w_{12}$};
  \node at (3,3) {$w_{13}$};
  \node at (3.6,3.6) {$w_{14}$};
  \node at (1.8,2.4) {$w_{15}$};
  \node at (4.2,2.4) {$w_{16}$};
  \node at (3,1.7) {$w_{17}$};
  \node at (1.2,1.4) {$w_{18}$};
  \node at (4.8,1.4) {$w_{21}$};
  \node at (2.1,1.2) {$w_{19}$};
  \node at (3.9,1.2) {$w_{20}$};
  \node at (1.9,0.4) {$w_{22}$};
  \node at (4.1,0.4) {$w_{24}$};
  \node at (3,0.4) {$w_{23}$};
  \end{tikzpicture}
  \end{center}
\end{example}

\begin{remark}[Hypergraph models]
	\label{rem:hypergraph-models}
An equivalent presentation of the notion of generalized simplicial model (\cref{def:generalized-simplicial-model}) is obtained via the notion of hypergraph.
A \emph{hypergraph} is a pair $\langle V,E \rangle$ where $V$ is a set of vertices, and $E \subseteq \Pow{V}$ is a set of hyperedges. Thus, essentially, it is the same data as a simplicial complex, except that the set $E$ of hyperedges is not required to be downward-closed. Indeed, simplicial complexes are a special case of hypergraphs. In a general hypergraph, we lose the geometric intuition of having higher-dimensional cells ($n$-simplexes) that can share a common sub-simplex; instead, we think of a hyperedge simply as a relation linking $n$ vertices together.

With that in mind, we can reformulate \cref{def:generalized-simplicial-model} as follows.
A \emph{hypergraph model} is a tuple $\cH = \la V, E, \chi, \ell \ra$, where $\langle V,E \rangle$ is a hypergraph, $\chi : V \to A$ is a colouring of the vertices such that every hyperedge has distinct colours, and $\ell : E \to \Pow{\AP}$ labels each hyperedge with a set of atomic proposition.
Notice that compared to generalized simplicial models, we got rid of one piece of data, the set of distinguished worlds~$W$.

It is immediate to see that any generalized simplicial model $\C = \la V, S, \chi, W, \ell \ra$ can be turned into a hypergraph model $\la V, W, \chi, \ell \ra$ where the set of hyperedges is $W$. Conversely, any hypergraph model $\cH = \la V, E, \chi, \ell \ra$ can be turned into a generalized simplicial model $\la V, \dclosure{E}, \chi, E, \ell \ra$, where the set of simplexes is the downward-closure of~$E$, and the distinguished set of worlds is~$E$ itself.
This correspondence is bijective due to the fact that we always require the set~$W$ to contain all the facets of~$\C$.
Thus, both notions are equivalent, and simply the result of a slight change of vocabulary. In this paper, we prefer to keep the geometric intuition of simplicial complexes, at the price of keeping track of the extra set~$W$.

Hypergraph models are discussed more thoroughly in \cite{hypergraph}, to study a different epistemic logic where formulas are separated into several sorts: ``agent formulas' and ``world formulas''.
\end{remark}

\subsection{Reasoning about alive and dead agents}

Until now, we discussed agents being ``alive'' or ``dead'' as a meta-level property of the model.
It is a natural idea to try to internalise this notion in the logic, i.e., to have formulas expressing whether an agent is alive or dead, such as ``agent $a$ knows that agent $b$ is dead''.
Fortunately, such formulas can already be expressed in our logic without any extra syntax, as derived operators:
$$
\deadprop{a} \,:=\, K_a\, \false \qquad \qquad
\aliveprop{a} \,:=\, \neg \deadprop{a}
$$
It is easy to check that the semantics of these formulas is, as expected (cf.\ \cref{def:alive}):
$$ \C,w \models \aliveprop{a} \quad \text{iff} \quad a \in \chi(w) $$

\begin{example}
Interestingly, that means we can investigate some structural properties of the models, without referring to the atomic propositions.
We illustrate this with the two distributed computing models of \cref{ex:message-passing,ex:immediate-snapshot}.
\begin{itemize}
\item The simplicial model depicted in in \cref{ex:message-passing}, where the set of worlds $W = \{w_0, \ldots, w_9\}$ contains only the facets, is an example of a model with \emph{detectable crashes}. That is, in every world where some agent is dead, one of the alive agents is aware of that fact.
Let~$\C_{\ref{ex:message-passing}}$ denote the simplicial model of \cref{ex:message-passing}.
Then for instance, in world $w_1$, only the agents~$a$ and~$b$ are alive, i.e., $\C_{\ref{ex:message-passing}}, w_1 \models \aliveprop{a} \land \aliveprop{b} \land \deadprop{c}$.
Moreover, $a$ does not know that $c$ is dead, but $b$ knows it: $\C_{\ref{ex:message-passing}}, w_1 \models \neg K_a\, \deadprop{c} \land K_b\,\deadprop{c}$.
\item The simplicial model $\C_{\ref{ex:immediate-snapshot}}$ of \cref{ex:immediate-snapshot} is neither minimal nor maximal. Thus, the agents might be aware or not of which other agents are alive, depending on where we are in the model.
In world $w_{17}$ in the centre of the picture, all three agents are alive, and they know that everyone is alive. For instance: $\C_{\ref{ex:immediate-snapshot}}, w_{17} \models K_a\,\aliveprop{b} \land K_a\,\aliveprop{c}$.
Closer to the border, in world $w_{15}$, agent~$a$ still knows that~$b$ is alive, but considers possible that~$c$ might be dead: $\C_{\ref{ex:immediate-snapshot}}, w_{15} \models K_a\,\aliveprop{b} \land \neg K_a\,\aliveprop{c}$.
In the worlds that contain the top vertex $w_0$, agent~$a$ considers possible that everyone might be dead: $\C_{\ref{ex:immediate-snapshot}}, w_{12} \models \neg K_a\,\aliveprop{b} \land \neg K_a\,\aliveprop{c}$.
\end{itemize}
\end{example}


\noindent
For groups of agents, we also use the following abbreviations: 
$$
\deadprop{B} \,:=\, \bigwedge_{a \in B} \deadprop{a} \qquad \qquad
\aliveprop{B} \,:=\, \neg D_B \false
$$
meaning that all the agents in $B$ are dead (resp., alive).
Note that $\deadprop{B}$ is not equivalent to $D_B\, \false$: the formula $D_B\, \false$ is true when at least one agent $a \in B$ is dead.

\subsection{Axiomatization: $\KBfour$ and beyond}
\label{sec:Axiomatization}

Simplicial models satisfy all the usual axioms of multi-agent epistemic logic, except for the axiom of \emph{truth}.
The logic we get is called $\KBfour$, and comprises the following axioms:
\begin{align*}
&\makebox[0.8cm]{\textbf{K:}\hfill} D_B (\phi \Rightarrow \psi) \Rightarrow (D_B\,\phi \Rightarrow D_B\,\psi)\\
&\makebox[0.8cm]{\textbf{B:}\hfill} \phi \Rightarrow D_B \neg D_B \neg \phi\\
&\makebox[0.8cm]{\textbf{4:}\hfill} D_B\,\phi \Rightarrow D_B D_B\,\phi
\end{align*}

\noindent
It is well known that Axiom $\mathbf{5}$ is provable in $\KBfour$ (see e.g.~\cite{sep-logic-modal}), so we also have:
\begin{align*}
\makebox[0.8cm]{\textbf{5:}\hfill} \neg D_B\,\phi \Rightarrow D_B \neg D_B\,\phi
\end{align*}

\noindent
The difference between $\KBfour$ and the more standard multi-agent epistemic logics $\Sfive$ is that we do 
not necessarily have axiom $\mathbf{T}$: $K_a\,\phi \Rightarrow \phi$.
Indeed, in any world of a simplicial model where the agent~$a$ is dead, axiom~$\mathbf{T}$ will be violated, since $K_a\,\false$ is satisfied.
Here are a few examples of valid formulas in $\KBfour$, related to the life and death of agents.

\begin{itemize}
\item \makebox[5.5cm]{Dead agents know everything:\hfill} $\KBfour \vdash \deadprop{a} \Rightarrow K_a\,\phi$.\\
\makebox[5.5cm]{More generally, for any $a \in B$:\hfill} $\KBfour \vdash \deadprop{a} \Rightarrow D_B\,\phi$. 
\item \makebox[5.5cm]{Alive agents satisfy Axiom \textbf{T}:\hfill} $\KBfour \vdash \aliveprop{a} \Rightarrow (K_a\,\phi \Rightarrow \phi)$.\\
\makebox[5.5cm]{More generally:\hfill} $\KBfour \vdash \aliveprop{B} \Rightarrow (D_B\,\phi \Rightarrow \phi)$. 
\item \makebox[5.5cm]{Alive agents know they are alive:\hfill} $\KBfour \vdash \aliveprop{a} \Rightarrow K_a\,\aliveprop{a}$.\\
\makebox[5.5cm]{More generally:\hfill} $\KBfour \vdash \aliveprop{B} \Rightarrow D_B\,\aliveprop{B}$.
\end{itemize}

%

We also consider six additional axioms that are not provable in $\KBfour$. The first one called \emph{monotonicity} is standard when dealing with distributed knowledge.
The second axiom, called \emph{union}, arises from the interplay between distributed knowledge and the possibility of having dead agents. It ensures that each world has a unique maximal set of alive agents, making the dead/alive status of individual agents, rather than groups, the primary concern.
The third axiom, \emph{non-emptiness}, says that every world has at least one agent that is alive.
The fourth one, the axiom of \emph{properness}, says that if two worlds have the same set of alive agents and no alive agent can distinguish them, then they must satisfy the same formulas.
It is best understood when taking $B = A$, in which case it says that in the worlds where everyone is alive, $\phi \Rightarrow D_A\,\phi$ holds.
The last two axioms, \emph{minimality} and \emph{maximality}, capture the sub-classes of minimal and maximal simplicial models, respectively.
They are explained in more detail in \cref{ex:min-max-axioms} below. In the following, we denote by $\complement{B}$ the complement of the set of agents $B$, i.e., $\complement{B} = A \setminus B$.
\begin{align*}
&\makebox[1.5cm]{\textbf{Mono:}\hfill}
  \makebox[8.8cm]{$D_B\, \phi \Rightarrow D_{B'}\,\phi$ \hfill}
  \text{for all } B \subseteq B' \subseteq A\\
&\makebox[1.5cm]{\textbf{Union:}\hfill}
  \makebox[8.8cm]{$\aliveprop{B} \land \aliveprop{B'} \Rightarrow \aliveprop{B\cup B'}$ \hfill}
  \text{for all } B, B' \subseteq A \\
&\makebox[1.5cm]{\textbf{NE:}\hfill}
  \textstyle \bigvee_{a \in A} \aliveprop{a}\\
&\makebox[1.5cm]{\textbf{P:}\hfill}
  \makebox[8.8cm]{$\aliveprop{B} \land \deadprop{\complement{B}} \land \phi \Rightarrow D_B(\deadprop{\complement{B}} \Rightarrow \phi)$\hfill}
  \text{for all } B \subseteq A\\
&\makebox[1.5cm]{\textbf{Min:}\hfill}
  \makebox[8.8cm]{$\aliveprop{B} \land \deadprop{\complement{B}} \Rightarrow D_B\,\deadprop{\complement{B}}$ \hfill}
  \text{for all } B \subsetneq A\\
&\makebox[1.5cm]{\textbf{Max:}\hfill}
  \makebox[8.8cm]{$\aliveprop{B} \Rightarrow \neg D_B \neg \deadprop{\complement{B}}$ \hfill}
  \text{for all } B \subsetneq A
t\end{align*}

\begin{example} \label{ex:min-max-axioms}
We illustrate the axioms \textbf{Min} and \textbf{Max} with the three models below, denoted by $\C_4, \C_5, \C_6$ from left to right. Notice that $\C_4$ is minimal, $\C_6$ is maximal, and $\C_5$ is neither minimal nor maximal.
\begin{center}
	\begin{tikzpicture}[auto,scale=1.2,cloudgrey/.style={draw=black,thick,circle,fill=cyan,inner sep=1pt,minimum size=11pt},cloud/.style={draw=black,thick,circle,fill=magenta!60,inner sep=1pt,minimum size=11pt}, cloudblack/.style={draw=black,thick,circle,fill=green,inner sep=1pt,minimum size=11pt}]
	\draw[thick, draw=black, fill=blue!60, fill opacity=0.15]
	  (4,0) -- (5,-0.577) -- (5,0.577) -- cycle;
	\draw[thick, draw=black, fill=blue!60, fill opacity=0.15]
	  (7.154,0) -- (6.154,-0.577) -- (6.154,0.577) -- cycle;
	\node[cloudblack] (b1) at (4,0) {$c$};
	\node[cloudblack] (b2) at (7.154,0) {$c$};
	\node[cloud] (g1) at (5,-0.577) {$b$};
	\node[cloudgrey] (w1) at (5,0.577) {$a$};
	\node[cloud] (g2) at (6.154,0.577) {$b$};
	\node[cloudgrey] (w2) at (6.154,-0.577) {$a$};
	\node at (4.65,0) {$w_1$};
	\node at (6.504,0) {$w_3$};
	\draw[thick] (w1) edge node[above] {$w_2$} (g2);
	\end{tikzpicture}
	\hspace{1.2cm}
	\begin{tikzpicture}[auto,rotate=30,scale=1.2,cloudgrey/.style={draw=black,thick,circle,fill=cyan,inner sep=1pt,minimum size=11pt},cloud/.style={draw=black,thick,circle,fill=magenta!60,inner sep=1pt,minimum size=11pt}, cloudblack/.style={draw=black,thick,circle,fill=green,inner sep=1pt,minimum size=11pt}]
	\draw[thick, draw=black, fill=blue!60, fill opacity=0.15]
	  (4,0) -- (5,-0.577) -- (5,0.577) -- cycle;
	\node at (4.65,0) {$w_1$};
	\node[cloudgrey] (b1) at (4,0) {$a$};
	\node[cloudblack] (g1) at (5,-0.577) {$c$};
	\node[cloud] (w1) at (5,0.577) {$b$};
	\draw[thick] (b1) edge node {$w_2$} (w1);
	\draw[thick] (g1) edge node[text opacity=0] {$w_0$} (b1); 
	\end{tikzpicture}
	\hspace{1.8cm}
	\begin{tikzpicture}[auto,rotate=30,scale=1.2,cloudgrey/.style={draw=black,thick,circle,fill=cyan,inner sep=1pt,minimum size=11pt},cloud/.style={draw=black,thick,circle,fill=magenta!60,inner sep=1pt,minimum size=11pt}, cloudblack/.style={draw=black,thick,circle,fill=green,inner sep=1pt,minimum size=11pt}]
	\draw[thick, draw=black, fill=blue!60, fill opacity=0.15]
	  (4,0) -- (5,-0.577) -- (5,0.577) -- cycle;
	\node at (4.65,0) {$w_1$};
	\node[cloudgrey,label=left:{$w_5$}] (b1) at (4,0) {$a$};
	\node[cloudblack,label=right:{$w_7$}] (g1) at (5,-0.577) {$c$};
	\node[cloud,label=above:{$w_6$}] (w1) at (5,0.577) {$b$};
	\draw[thick] (w1) edge node {$w_3$} (g1);
	\draw[thick] (b1) edge node {$w_2$} (w1);
	\draw[thick] (g1) edge node {$w_4$} (b1);
	\end{tikzpicture}
	\[
	\begin{array}{@{\hskip 1cm}l@{\hskip 0.8cm}l@{\hskip 0.8cm}l}
	\C_4,w_2 \models \neg K_a\,\deadprop{c} &
	\C_5,w_2 \models \neg D_{\{a,b\}}\,\deadprop{c} &
	\C_6,w_1 \models \neg D_{\{b,c\}} \neg \deadprop{a} \\
	\C_4,w_2 \models \neg K_b\,\deadprop{c} &
	\C_5,w_1 \models D_{\{b,c\}} \neg \deadprop{a} &
	\C_6,w_1 \models \neg D_{\{a\}}\neg \deadprop{\{b,c\}} \\
	\C_4,w_2 \models D_{\{a,b\}}\,\deadprop{c} &
	 &
	\end{array}
	\]
\end{center}
\begin{itemize}
\item Axiom \textbf{Min} can be understood intuitively as saying that crashes must be \emph{detectable} (cf.\ \cref{ex:message-passing}). Indeed, it says that whenever some set of agents $\complement{B}$ have crashed, there is distributed knowledge among the remaining agents that they have crashed.
This can be seen in world~$w_2$ of model $\C_4$, where the set of alive agents is $B = \{a,b\}$. Neither~$a$ not~$b$, individually, know that agent~$c$ is dead. However, there is distributed knowledge among $\{a,b\}$ that~$c$ is dead.
Thus, Axiom \textbf{Min} is valid in model $\C_4$.
The way to invalidate Axiom \textbf{Min} is to have a world which is a sub-world of another, such as world~$w_2$ in model~$\C_5$. There, we do not have $D_{\{a,b\}}\,\deadprop{c}$, because of the possibility of world~$w_1$ where~$c$ is alive.
\item Axiom \textbf{Max} says, intuitively, that all crash patterns are possible and undetectable.
Thus, all the sub-worlds always exist.
For example, in world~$w_1$ of model~$\C_6$, all agents are alive. Any subset~$B$ of the alive agents considers possible that everyone else might be dead. Particular instances of Axiom \textbf{Max} for $B = \{b,c\}$ and $B = \{a\}$ are written below the picture. The first one is satisfied because of the existence of world~$w_3$; the second one, because of world~$w_5$. Thus, Axiom \textbf{Max} is valid in model $\C_6$.
The way to invalidate Axiom \textbf{Max} is to have a missing sub-world, such as in model~$\C_5$.
Since the $bc$-coloured edge is not a world of~$\C_5$, Axiom \textbf{Max} fails for $B = \{b,c\}$.
\end{itemize}
\end{example}

\begin{remark}
In the conference version of this paper~\cite{stacs22}, we had a different set of axioms. This is due to two facts: (i) we only considered standard knowledge $K_a\,\phi$ instead of distributed knowledge, and (ii) we worked with the sub-class of minimal models only, rather that the full generality presented here (see \cref{rem:generalized-models}).
Moreover, the original version~\cite{stacs22} was missing Axiom~\textbf{P}. This was fixed later on arXiv~\cite{stacs22-arxiv}.
Thus, because of (i), only the case of a single agent $B = \{a\}$ is required; and because of (ii), the two axioms \textbf{P} and \textbf{Min} are merged into a single axiom called \textbf{SA}.
With those two assumptions in mind, we can check that our axioms are indeed consistent with the one of~\cite{stacs22-arxiv},
$\textbf{SA}_\mathbf{a}\!\!: \aliveprop{a} \land \deadprop{\complement{\{a\}}} \land \phi \Rightarrow K_a\,\phi$.
\end{remark}

\begin{remark}
The Axioms $\textbf{Min}$ and $\textbf{Max}$ for $B = A$ are vacuously true.
For Axiom~$\textbf{P}$, we only really need the instances where $\phi = p \in \AP$ is an atomic proposition.
\end{remark}

One can check that \textbf{Mono}, \textbf{Union}, \textbf{NE} and \textbf{P} are valid in all simplicial models.
The axiom \textbf{Min} is valid (exactly) in all minimal simplicial models; and the axiom \textbf{Max} is valid (exactly) in all maximal simplicial models.
Hence, let us write $\SC$ (``the logic of simplicial complexes'') for the proof system given by the axioms $\KBfour + \textbf{Mono} + \textbf{Union} + \textbf{NE} + \textbf{P}$, as well as all propositional tautologies, closure by modus ponens, and the necessitation rule: if $\phi$ is a tautology, then $D_B\,\phi$ is a tautology.
We also write $\SCmin$ and $\SCmax$ for the proof system $\SC$ augmented, respectively, with the axioms \textbf{Min} and \textbf{Max}.

\begin{proposition}
\label{prop:soundness}
The proof system $\SC$ (resp., $\SCmin$, $\SCmax$) is sound with respect to the class of generalized (resp.\ minimal, maximal) simplicial models.
\end{proposition}
\begin{proof}
The proof of soundness is straightforward as usual by induction on the proof of a formula~$\phi$. We only check that the axioms of $\SC$ are valid in all simplicial models.

Let $\C = \la V,S,\chi,W,\ell \ra$ be a simplicial model.
The axioms $\textbf{K}$, $\textbf{B}$ and $\textbf{4}$ hold because the satisfaction relation on simplicial models is a Kripke-style semantics in disguise (see \cref{sec:kripke}, in particular \cref{thm:truth-equiv}). For now, let us give a direct proof for Axiom $\textbf{4}$. Let $w \in W$ be a world of $\C$ and assume that $\C,w \models D_B\,\phi$.
In order to show that $\C,w \models D_B D_B\,\phi$, let $w' \in W$ such that $B \subseteq \chi(w \cap w')$ and let $w'' \in W$ such that $B \subseteq \chi(w' \cap w'')$.
Since $\C$ is a chromatic simplicial complex, each colour appears at most once in a simplex. So, in fact, we have $B \subseteq \chi(w \cap w' \cap w'')$, and in particular $B \subseteq \chi(w \cap w'')$. Since we assumed that $\C,w \models D_B\,\phi$, we obtain $\C,w'' \models \phi$ as required.

The proof is similar for Axiom $\textbf{Mono}$: assume that $\C,w \models D_B\,\phi$ and that $B \subseteq B'$.
To show $\C,w \models D_{B'}\,\phi$, let $w' \in W$ such that $B' \subseteq \chi(w \cap w')$. Then $B \subseteq \chi(w \cap w')$, so by assumption $\C,w' \models \phi$.
In the same vein, $\textbf{Union}$ follows from the fact that if the vertices of a simplex are colored by colors $B$ and $B'$, then they are colored by $B\cup B'$. That is, if $\C, w\models \aliveprop{B}$ and $\C, w\models \aliveprop{B'}$, then $\C, w\models \aliveprop{B\cup B'}$.
The validity of $\textbf{NE}$ comes from the fact that a simplex is always non-empty (see \cref{def:simplicial-complex}). So for any $w \in W$, there is at least one vertex $v \in w$. Then for $a := \chi(v)$, we have $\C,w \models \aliveprop{a}$.
Validity of Axiom $\textbf{P}$ is a bit more involved. Assume that $\C,w \models \aliveprop{B} \land \deadprop{\complement{B}} \land \phi$, i.e., $\C,w \models \phi$ and moreover the set of colours of the vertices of~$w$ is exactly~$B$.
To prove that $\C,w \models D_B(\deadprop{\complement{B}} \Rightarrow \phi)$, let $w' \in W$ such that $B \subseteq \chi(w \cap w')$. 
So $w'$ contains all the vertices of $w$.
If we assume moreover that $\C,w' \models \deadprop{\complement{B}}$, then $w'$ cannot contain any extra vertex, i.e.\ $w = w'$.
Thus we must have $\C,w' \models \phi$, which concludes the proof.

To show that Axiom $\textbf{Min}$ is valid in every minimal simplicial model is very similar to the one of Axiom~$\textbf{P}$ above. Indeed, assume that the set of colours in $w \in W$ is exactly~$B$.
Since $w$ is a facet, the only possible $w' \in W$ such that $B \subseteq \chi(w \cap w')$ is $w' = w$ itself.
In particular, we do have $\C,w' \models \deadprop{\complement{B}}$.

Finally, for Axiom $\textbf{Max}$, assume that $w \in W$ contains at least the colours in $B$. Let $w' \subseteq w$ be the sub-simplex of~$w$ whose colours are exactly those of $B$ (potentially, $w' = w$).
Since~$\C$ is a maximal model, we must have $w' \in W$; and moreover $\C,w' \models \deadprop{\complement{B}}$.
So we have $\C,w \models \neg D_B \neg \deadprop{\complement{B}}$ as required.
\end{proof}

Completeness also holds for $\SC$, $\SCmin$ and $\SCmax$, but the proof is more intricate.
Indeed, we will use a detour via an equivalence with Kripke models, that we develop in \cref{sec:kripke}.
We then prove the three completeness results in \cref{sec:completeness}.

\section{Equivalent classes of Kripke models}
\label{sec:kripke}

In normal modal logics, whose semantics is based on Kripke models, there is a well-known correspondence between axioms of the logic and properties of the corresponding Kripke frames~\cite{sep-logic-modal}.
Namely, Axiom \textbf{K} holds in all Kripke models; while Axioms \textbf{B} and \textbf{4} are valid exactly on the class of Kripke models whose accessibility relation is symmetric and transitive, respectively.
So the logic $\KBfour$ is sound and complete with respect to the class of symmetric and transitive Kripke models.
This will be our starting point to define the class of Kripke models that is equivalent to simplicial models.
However, as we saw in \cref{sec:Axiomatization}, simplicial models have some additional built-in assumptions, that we need to impose on Kripke models too.
Crucially, since non-pure simplicial models do not obey Axiom \textbf{T}, we consider Kripke models whose accessibility relation is not necessarily reflexive.

\subsection{Partial epistemic models}
\label{sec:partial-epistemic-models}

Relations that are symmetric and transitive are called \emph{Partial Equivalence Relations} in the context of PER semantic models of programming languages. They also appear e.g.\ in~\cite{KPER}, where they are called ``Kripke logical partial equivalence relations''.

\begin{definition}
A \emph{Partial Equivalence Relation} (PER) on a set $X$ is a relation ${R \subseteq X \times X}$ that is symmetric and transitive (but not necessarily reflexive).
\end{definition}

The \emph{domain} of a PER $R$ is the set $\dom{R} = \{ x \in X \mid R(x,x)\} \subseteq X$, and it is easy to see that $R$ is an equivalence relation on its domain, and empty outside of it.
Thus, PERs are equivalent to the ``local equivalence relations'' defined in~\cite{Ditmarsch21}.
We now fix a set of agents~$A$.

\begin{definition}
A \emph{partial epistemic frame} $M = \la W,\sim \ra$ is a Kripke frame such that each relation $(\sim_a)_{a \in A}$ is a PER.
\end{definition}

We say that agent~$a$ is \emph{alive} in a world~$w$ when $w \in \dom{\sim_a}$, i.e., when $w \sim_a w$.
We write $\live{w}$ for the set of agents that are alive in world $w$. 
Finally, we say that a world $w$ is a \emph{sub-world} of $w'$ when $\live{w} \subsetneq \live{w'}$ and $w \sim_a w'$ for all $a \in \live{w}$.
We now define four properties of partial epistemic frames, echoing the four Axioms \textbf{NE}, \textbf{P}, \textbf{Min} and \textbf{Max} defined in \cref{sec:Axiomatization}.

\begin{definition}
\label{def:PER-frame-properties}
Let $M = \la W,\sim \ra$ be a partial epistemic frame. We say that:
\begin{itemize}
\item $M$ \emph{has no empty world} when $\live{w} \neq \emptyset$ for all $w \in W$.
\item $M$ is \emph{proper} when two distinct worlds with the same set of alive agents can always be distinguished by at least one alive agent. More formally, $M$ is proper when for every $w \neq w' \in W$ such that $\live{w} = \live{w'}$, there exists $a \in \live{w}$ such that $w \not \sim_a w'$.
\item $M$ is \emph{minimal} if it has no sub-world, i.e., for all $w,w' \in W$, if $\live{w} \subsetneq \live{w'}$, then there exists $a \in \live{w}$ such that $w \not \sim_a w'$.
\item $M$ is \emph{maximal} if it has all sub-worlds, i.e., for all $w' \in W$ and for all non-empty $B \subsetneq \live{w'}$, there exists $w \in W$ such that $\live{w} = B$ and $w \sim_a w'$ for all $a \in B$.
\end{itemize}
\end{definition}

\begin{remark} \label{rem:terminology-proper}
We have slightly changed our terminology compared to the conference version of this work.
Indeed, the property that we used to called ``proper'' in~\cite{stacs22} is actually equivalent to the conjunction ``non-empty and proper and minimal'' in the sense of \cref{def:PER-frame-properties}.
Since we are now interested in studying generalized simplicial models, and not just the subclass of minimal ones, we have refined this into three separate properties.
We believe that our new usage of the word ``proper'', which is now less specific than the one of~\cite{stacs22}, better captures what is usually understood as proper in the context of \textbf{S5} epistemic frames.
\end{remark}

\begin{example} 
\label{ex:proper-minimal-maximal}
Four partial epistemic frames over the set of agents $A = \{a,b,c\}$ are represented below.
The frame at the top left is not proper, while the three other frames are proper.
The frame at top right is neither minimal nor maximal: the world $w_3$ is a sub-world of $w_2$, but not all possible sub-worlds of $w_2$ exist.
The frame at the bottom left is minimal: neither $w_4$ nor $w_5$ is a sub-world of the other.
The frame at the bottom right is maximal: both $w_6$ and $w_8$ are sub-worlds of $w_7$, and there can be no other sub-world without breaking properness.

\begin{center}
\begin{tabular}{c@{\hskip 2cm}c}
\begin{tikzpicture}[auto]
\tikzset{every loop/.style={}}
\node (p) at (-2,0) {$w_0$};
\node (q) at (0,0) {$w_1$};
\path (p) edge[loop above] node {$a,b$} (p)
      (p) edge node {$a,b$} (q)
      (q) edge[loop above] node {$a,b$} (q);
\end{tikzpicture}
&
\begin{tikzpicture}[auto]
\tikzset{every loop/.style={}}
\node (p) at (-2,0) {$w_2$};
\node (q) at (0,0) {$w_3$};
\path (p) edge[loop above] node {$a,b,c$} (p)
      (p) edge node {$a,b$} (q)
      (q) edge[loop above] node {$a,b$} (q);
\end{tikzpicture}
\\[12pt]
\begin{tikzpicture}[auto]
\tikzset{every loop/.style={}}
\node (p) at (-2,0) {$w_4$};
\node (q) at (0,0) {$w_5$};
\path (p) edge[loop above] node {$a,b$} (p)
      (p) edge node {$a$} (q)
      (q) edge[loop above] node {$a,c$} (q);
\end{tikzpicture}
&
\begin{tikzpicture}[auto]
\tikzset{every loop/.style={}}
\node (p) at (-2,0) {$w_6$};
\node (q) at (0,0) {$w_7$};
\node (r) at (2,0) {$w_8$};
\path (p) edge[loop above] node {$a$} (p)
      (p) edge node {$a$} (q)
      (q) edge[loop above] node {$a,b$} (q)
      (q) edge node {$b$} (r)
      (r) edge[loop above] node {$b$} (r);
\end{tikzpicture}
\end{tabular}
\end{center}
\end{example}

\begin{definition}
\label{def:partial-epistemic-model}
A \emph{partial epistemic model}  $\model = \la W,\sim,L \ra$  over the set of agents $A$ consists of a partial epistemic frame 
 $\la W, \sim \ra$ together with function $L : W \to \Pow{\AP}$.
%
\end{definition}
Intuitively, $L(w)$ is the set of atomic propositions that are true in the world $w$.
Note that partial epistemic models are simply Kripke models (in the usual sense of normal modal logics), such that all the accessibility relations $(\sim_a)_{a \in A}$ are PERs.
Thus, we can straightforwardly define the semantics of an epistemic formula $\phi\in \lang_D$ in such a model.
Formally, given a pointed partial epistemic model $(\model,w)$,
we define by induction on~$\phi$ the \emph{satisfaction relation} $\model,w \models \phi$ as follows:
\[ \begin{array}{lcl}
\model,w \models p & \text{iff} & p \in L(w) \\
\model,w \models \neg\phi & \text{iff} & \model,w \not\models \phi \\
\model,w \models \phi\et\psi & \text{iff} & \model,w \models \phi \text{ and } \model,w \models \psi \\
\model,w \models D_B\,\phi & \text{iff} & \model,w' \models \phi \text{ for all } w' \text{ such that } w \sim_B w'
\end{array} \]
where $\sim_B$ is the intersection of the relations $(\sim_a)_{a \in B}$, i.e., $w \sim_B w'$ iff $w \sim_a w'$ for all $a \in B$.

\subsection{Relating simplicial models and partial epistemic models}
\label{sec:model-equivalence}

In this section, we show how to canonically associate a proper partial epistemic frame with any chromatic simplicial complex, and vice versa.
More precisely, for any generalized simplicial model $\C$, we construct an associated partial epistemic model $\kappa(\C)$ which is proper and has no empty world.
Conversely, for any proper partial epistemic model $M$ that has no empty world, we associate a generalized simplicial model $\sigma(M)$.
We also show that in both cases, the notions of ``minimal'' and ``maximal'' models are preserved.
In \cref{thm:truth-equiv}, we state the key property of this section: the two maps $\kappa$ and $\sigma$ preserve the satisfaction relation.
Similar correspondences appears in~\cite{gandalf-journal,Ditmarsch21,stacs22}, with some differences:
\begin{itemize}
\item Here, we work in a more general framework of \emph{generalized simplicial models}, rather than just the minimal ones where worlds are facets.
As a consequence, the corresponding class of Kripke models is larger (because we changed the meaning of ``proper'', cf.\ \cref{rem:terminology-proper}).
\item On the other hand, here we are not concerned with the morphisms between models. Thus, we do not prove that $\sigma$ and $\kappa$ form a categorical equivalence between simplicial models and Kripke models.
For instance, while~\cite{gandalf-journal} and~\cite{stacs22} show that $\kappa \circ \sigma(M)$ is isomorphic to $M$; and~\cite{Ditmarsch21} 
shows that $\kappa \circ \sigma(M)$ is bisimilar to $M$; here we only prove that the satisfaction relation is preserved, which is the weakest of those three properties.
Nonetheless, it will be sufficient for our purpose, the completeness results of \cref{sec:completeness}.
\end{itemize}

\begin{definition}
\label{def:adjF}
Let $\C = \la V, S, \chi, W, \ell \ra$ be a generalized simplicial model on the set of agents~$A$ and atomic propositions~$\AP$. Its associated partial epistemic model is $\kappa(\C)=\la W, \sim, L \ra$, whose set of worlds $W$ is the same as the one of~$\C$, and whose relation~$\sim_a$, for each agent $a \in A$,
is given by $w \sim_a w'$ iff $a \in \chi(w \cap w')$.
The labelling is simply defined as $L(w) = \ell(w)$.
\end{definition}

\begin{proposition}
\label{lemma:proper}
$\kappa(\C)$ is a proper partial epistemic frame that has no empty world.
\end{proposition}
\begin{proof}
The relation $\sim_a$ is easily seen to be a symmetric and transitive, because since the simplicial complex $\C$ is chromatic, there can be at most one vertex $v \in w \cap w'$ with $\chi(v) = a$.
Moreover, since the worlds of $\kappa(\C)$ are simplexes of $\C$, and simplexes are always non-empty by definition, we immediately see that $\kappa(\C)$ has no empty world.
Finally, to show that $\kappa(\C)$ is proper, consider two distinct worlds $w$ and $w'$ in $\kappa(\C)$, i.e., two simplexes of $\C$, and assume they have the same set of alive agents, i.e., $\chi(w) = \chi(w')$.
Since a simplex is uniquely determined by its set of vertices, there is at least one vertex of $w$, say $v$, that does not belong to $w'$ (otherwise we would have $w = w'$).
Let $a = \chi(v)$ be the colour of~$v$.
Then $a$ is alive in~$w$ because $a \in \chi(w \cap w)$;
and $w \not \sim_a w'$ because $v \not \in w \cap w'$ and there can be only one vertex with colour $a$ in $w$.
\end{proof}

\begin{proposition}
\label{lemma:min-max-kappa}
If $\C$ is minimal (resp., maximal), then $\kappa(\C)$ is minimal (resp., maximal).
\end{proposition}
\begin{proof}
Assume the simplicial model $\C$ is minimal, i.e., that all worlds are facets: ${W = \Facets(\C)}$.
Let $w,w'$ be two worlds of $\kappa(\C)$ with $\live{w} \subsetneq \live{w'}$.
Then there must be at least one vertex of $w$, say $v$, that does not belong to $w'$: otherwise we would have $w \subsetneq w'$, which contradicts the fact that $w$ is a facet.
Let $a = \chi(v)$ be the colour of~$v$; then $a$ is alive in~$w$ and $w \not \sim_a w'$.
Thus $\kappa(\C)$ is minimal.

Assume now that $\C$ is maximal, i.e., all simplexes are worlds: $W = S$.
Let $w'$ be a world of $\kappa(\C)$, whose set of alive agents is $\live{w'} = \chi(w')$.
Let $B \subsetneq \chi(w')$ be a non-empty subset of alive agents, and let $w \subsetneq w'$ be the face of $w'$ that consists of all vertices whose colour is in~$B$.
Then we have $w \in W$ (because all simplexes are worlds), and it is easy to check that $w$ is a sub-world of $w'$ in $\kappa(\C)$.
So $\kappa(\C)$ is maximal.
\end{proof}

Conversely, we now consider a partial epistemic model $M=\la W,\sim,L \ra$,
and we define the associated simplicial model $\sigma(M)$.
Intuitively, each world $w \in W$ where $k$ agents are alive will be represented by a simplex of dimension~$k-1$, whose vertices are coloured by $\live{w}$.
These simplexes must then be ``glued'' together according to the indistinguishability relations.
Formally, this is done by a quotient construction, described in \cref{def:adjG} below.
When $a$ is alive in a world $w$, we write $[w]_a$ for the equivalence class of $w$ w.r.t.\ $\sim_a$, within $\dom{\sim_a}$.

\begin{definition}
\label{def:adjG}
Let $M = \la W, \sim,L \ra$ be a proper partial epistemic model with no empty world.
Its associated chromatic simplicial complex is $\sigma(M) = \la V, S, \chi, \widehat{W}, \ell \ra$, where:
\begin{itemize}
\item The set of vertices is $V = \{(a,[w]_a) \mid w \in W, a \in \live{w} \}$. We denote such a vertex $(a,[w]_a)$ by $v^w_{a}$ for succinctness; but note that $v^w_{a} = v^{w'}_{a}$ whenever $w \sim_a w'$.
\item The set $S$ of simplexes is generated by sets of the form $X_w = \{ v^w_{a} \mid a \in \live{w}\}$ for each $w \in W$; as well as all their sub-simplexes.
\item The colouring is given by $\chi(v^w_a) = a$.
\item The set of worlds is $\widehat{W} = \{ X_w \mid w \in W \}$.
\item The labelling is $\ell(X_w) = L(w)$.
\end{itemize}
\end{definition}

\begin{proposition}
\label{prop:sigma-well-defined}
$\sigma(M)$ is indeed a generalized simplicial model. Moreover, if $M$ is minimal (resp., maximal) then $\sigma(M)$ is minimal (resp., maximal).
\end{proposition}
\begin{proof}
The set $S$ of simplexes is downward-closed by construction, and every singleton $\{v^w_a\}$ belongs to $S$ since it is a sub-simplex of $X_w$.
All vertices of $X_w$ have distinct colours by construction, so $\la V, S, \chi \ra$ is indeed a chromatic simplicial complex.
We still need to show that $\widehat{W}$ contains all facets. That is also true by construction, since every simplex is a sub-simplex of some $X_w$, and a facet can only be a sub-simplex of itself.
Lastly, for the labelling to be well-defined, we need to make sure that $X_w \neq X_{w'}$ whenever $w \neq w'$, i.e., that there is a bijection between $W$ and $\widehat{W}$.
Assume by contradiction that this is not the case: then $X_w = X_{w'}$ implies that $\live{w} = \live{w'}$, and that $v^w_a = v^{w'}_a$ for all $a \in \live{w}$.
This is not possible because we assumed that $M$ is proper.

Assume now that $M$ is minimal. We want to show that every $X_w \in \widehat{W}$ is a facet of $\sigma(M)$.
It suffices to show that for all $w \neq w'$, $X_w \not \subseteq X_{w'}$. 
Assume by contradiction that $X_w \subseteq X_{w'}$. We already proved that equality is impossible, so we must have $\live{w} \subsetneq \live{w'}$, and for every $a \in \live{w}$, $v^w_a = v^{w'}_a$.
This contradicts the minimality of~$M$.

Finally, assume instead that $M$ is maximal, and let us show that every simplex of $\sigma(M)$ belongs to $\widehat{W}$. Let $X$ be a simplex of $\sigma(M)$, that is, $X \subseteq X_w$ for some~$w$. We want to show that there exists $w' \in W$ such that $X = X_{w'}$.
Let $B = \chi(X)$ be the set of colours of the vertices of~$X$.
Since $B \subseteq \live{w}$, by maximality of~$M$, there exists a world $w' \in W$ such that $\live{w'} = B$ and $w' \sim_a w$ for all $a \in B$.
Then for every $a \in B$ we have $v^{w'}_a = v^w_a$, so $X = X_{w'}$. 
\end{proof}

We now check that the associated models given by $\kappa$ and $\sigma$ preserve the semantics of epistemic logic formulas in $\lang_D$.

\begin{theorem} \label{thm:truth-equiv}
Given a pointed simplicial model $(\C,w)$, we have $\C,w \models \varphi$ iff $\kappa(\C),w \models \varphi$.
Conversely, given a pointed partial epistemic model $(M,w)$ which is proper and has no empty world, we have $M,w \models \varphi$ iff $\sigma(M),X_w \models \varphi$.
\end{theorem}
\begin{proof}
The first equivalence is straightforward by induction on the structure of the formula~$\phi$.
Indeed, the base case of atomic propositions comes from the fact that we keep the same labelling $L(w) = \ell(w)$ in \cref{def:adjF}. The cases of the operators $\land$ and $\neg$ are obvious using the induction hypothesis.
And for a formula of the form $D_B\,\phi$, notice that since we defined $w \sim_a w'$ iff $a \in \chi(w \cap w')$ in \cref{def:adjF}, we also get $w \sim_B w'$ iff $B \subseteq \chi(w \cap w')$.

The other half of the theorem is also proved by induction on the formula~$\phi$. Atomic propositions, conjunction and negation are straightforward.
For a formula of the form $D_B\,\phi$, all we have to show is that $w \sim_B w'$ in the model $M$ iff $B \subseteq \chi(X_w \cap X_{w'})$ in $\sigma(M)$.
This follows from the fact that $w \sim_a w'$ iff $v^w_a = v^{w'}_a$ iff $a \in \chi(X_w \cap X_{w'})$.
\end{proof}

\begin{example}
The toy model below depicts a simplicial model (left) with set of worlds $W = \{w_1, \ldots, w_7\}$. On the right is the equivalent partial epistemic model obtained by applying~$\kappa$.
Alternatively, one can also apply $\sigma$ to the model on the right in order to produce the simplicial model depicted on the left.
Some edges that can be deduced by transitivity have been omitted on the picture of the epistemic model on the right.
Notice that the set of alive agents in a world can be read directly from the reflexive loops.
In both models, the worlds~$w_5$ and $w_6$ are sub-worlds of $w_1$; and world $w_7$ is a sub-world of $w_4$.
\begin{center}
	\begin{tikzpicture}[auto,scale=1.2,cloudgrey/.style={draw=black,thick,circle,fill=cyan,inner sep=1pt,minimum size=11pt},cloud/.style={draw=black,thick,circle,fill=magenta!60,inner sep=1pt,minimum size=11pt}, cloudblack/.style={draw=black,thick,circle,fill=green,inner sep=1pt,minimum size=11pt}]
	\draw[thick, draw=black, fill=blue!60, fill opacity=0.15]
	  (4,0) -- (5,-0.577) -- (5,0.577) -- cycle;
	\draw[thick, draw=black, fill=blue!60, fill opacity=0.15]
	  (7.154,0) -- (6.154,-0.577) -- (6.154,0.577) -- cycle;
	\node[cloudblack] (b1) at (4,0) {$c$};
	\node[cloudblack, label=right:{$w_7$}] (b2) at (7.154,0) {$c$};
	\node[cloud] (g1) at (5,-0.577) {$b$};
	\node[cloudgrey] (w1) at (5,0.577) {$a$};
	\node[cloud] (g2) at (6.154,0.577) {$b$};
	\node[cloudgrey] (w2) at (6.154,-0.577) {$a$};
	\node at (4.65,0) {$w_1$};
	\node at (6.504,0) {$w_4$};
	\draw[thick] (w1) edge node[above] {$w_2$} (g2);
	\draw[thick] (w2) edge node[below] {$w_3$} (g1);
	\draw[thick] (w1) edge node[above, xshift=-4pt] {$w_5$} (b1);
	\draw[thick] (g1) edge node[below, xshift=-4pt] {$w_6$} (b1);
	\end{tikzpicture}
	\hspace{1.5cm}
	\begin{tikzpicture}[auto, scale=1.4, font={\small}, {every loop/.style}={looseness=2, min distance=2mm}]
	\node (w1) at (-1,0) {$w_1$};
	\node (w2) at (0,0.5) {$w_2$};
	\node (w3) at (0,-0.5) {$w_3$};
	\node (w4) at (1,0) {$w_4$};
	\node (w5) at (-2,0.5) {$w_5$};
	\node (w6) at (-2,-0.5) {$w_6$};
	\node (w7) at (2,0) {$w_7$};
	
	\path (w1) edge[in=110,out=70,loop] node[above] {\scriptsize $a,\!b,\!c$} (w1)
		  (w2) edge[in=110,out=70,loop] node[above] {\scriptsize $a,\!b$} (w2)
		  (w3) edge[in=110,out=70,loop] node[above] {\scriptsize $a,\!b$} (w3)
		  (w4) edge[in=110,out=70,loop] node[above] {\scriptsize $a,\!b,\!c$} (w4)
		  (w5) edge[in=110,out=70,loop] node[above] {\scriptsize $a,\!c$} (w5)
		  (w6) edge[in=110,out=70,loop] node[above] {\scriptsize $b,\!c$} (w6)
		  (w7) edge[in=110,out=70,loop] node[above] {\scriptsize $c$} (w7)
	      (w1) edge node[inner sep=1pt] {\scriptsize $a$} (w2)
	      (w1) edge node[inner sep=1pt,swap] {\scriptsize $b$} (w3)
	      (w1) edge node[inner sep=1pt] {\scriptsize $a,\!c$} (w5)
	      (w1) edge node[inner sep=1pt] {\scriptsize $b,\!c$} (w6)
	      (w2) edge node[inner sep=1pt] {\scriptsize $b$} (w4)
	      (w3) edge node[inner sep=1pt,swap] {\scriptsize $a$} (w4)
	      (w4) edge node[inner sep=1pt] {\scriptsize $c$} (w7);
	\end{tikzpicture}
\end{center}
\end{example}

\begin{example}
Recall the synchronous broadcast model with detectable crashes of \cref{ex:message-passing}. Its associated partial epistemic model is depicted below. Note that both models are minimal in the appropriate sense.

\begin{center}
	\begin{tikzpicture}[scale=1.9, auto, font={\small}, baseline={(0,0)}]
	\begin{pgfonlayer}{nodelayer}
		\node [style=cyan vertex] (0) at (0, 1) {$a$};
		\node [style=green vertex] (1) at (1, -0.5) {$c$};
		\node [style=magenta vertex] (2) at (-1, -0.5) {$b$};
		\node [style=green vertex] (3) at (1.5, 2) {$c$};
		\node [style=cyan vertex] (4) at (2.5, 0.5) {$a$};
		\node [style=magenta vertex] (5) at (1, -2.25) {$b$};
		\node [style=green vertex] (6) at (-1, -2.25) {$c$};
		\node [style=magenta vertex] (7) at (-1.5, 2) {$b$};
		\node [style=cyan vertex] (8) at (-2.5, 0.5) {$a$};
		\node [style=none] (12) at (-1, -0.5) {};
		\node [style=none] (13) at (0, 1) {};
		\node [style=none] (14) at (1, -0.5) {};
		\node [style=none] (15) at (0, 0) {$w_0$};
	\end{pgfonlayer}
	\begin{pgfonlayer}{edgelayer}
		\draw [style=thick] (2) to node [inner sep=2pt] {$w_3$} (8);
		\draw [style=thick] (8) to node [inner sep=2pt] {$w_2$} (7);
		\draw [style=thick] (7) to node [inner sep=1pt] {$w_1$} (0);
		\draw [style=thick] (0) to node [inner sep=1pt] {$w_9$} (3);
		\draw [style=thick] (3) to node [inner sep=2pt] {$w_8$} (4);
		\draw [style=thick] (4) to node [inner sep=2pt] {$w_7$} (1);
		\draw [style=thick] (1) to node [inner sep=2pt] {$w_6$} (5);
		\draw [style=thick] (5) to node [inner sep=2pt] {$w_5$} (6);
		\draw [style=thick] (6) to node [inner sep=2pt] {$w_4$} (2);
		\draw [style=filled path] (13.center)
			 to (14.center)
			 to (12.center)
			 to cycle;
	\end{pgfonlayer}
\end{tikzpicture}

	\hspace{1.5cm}
	\begin{tikzpicture}[auto, scale=1.4, font={\small}, {every loop/.style}={looseness=2, min distance=4mm}, baseline={(0,0)}]
	\begin{pgfonlayer}{nodelayer}
		\node [style=empty node] (0) at (0, 0) {$w_0$};
		\node [style=empty node] (1) at (2, 0) {$w_7$};
		\node [style=empty node] (2) at (1, 1.5) {$w_9$};
		\node [style=empty node] (3) at (1, -1.5) {$w_6$};
		\node [style=empty node] (4) at (-1, -1.5) {$w_4$};
		\node [style=empty node] (5) at (-1, 1.5) {$w_1$};
		\node [style=empty node] (6) at (-2, 0) {$w_3$};
		\node [style=empty node] (7) at (-3, 1.5) {$w_2$};
		\node [style=empty node] (8) at (3, 1.5) {$w_8$};
		\node [style=empty node] (9) at (0, -3) {$w_5$};
	\end{pgfonlayer}
	\begin{pgfonlayer}{edgelayer}
		\draw (0) to node [inner sep=1pt] {\scriptsize $a$} (5);
		\draw (5) to node [inner sep=1pt, swap] {\scriptsize $b$} (7);
		\draw (7) to node [inner sep=1pt, swap] {\scriptsize $a$} (6);
		\draw (6) to node [inner sep=1pt, swap] {\scriptsize $b$} (0);
		\draw (0) to node [inner sep=1pt, swap] {\scriptsize $b$} (4);
		\draw (4) to node [inner sep=1pt, swap] {\scriptsize $c$} (9);
		\draw (9) to node [inner sep=1pt, swap] {\scriptsize $b$} (3);
		\draw (3) to node [inner sep=1pt, swap] {\scriptsize $c$} (0);
		\draw (0) to node [inner sep=1pt, swap] {\scriptsize $c$} (1);
		\draw (1) to node [inner sep=1pt, swap] {\scriptsize $a$} (8);
		\draw (8) to node [inner sep=1pt, swap] {\scriptsize $c$} (2);
		\draw (2) to node [inner sep=1pt] {\scriptsize $a$} (0);
		\draw (6) to node [inner sep=1pt, swap] {\scriptsize $b$} (4);
		\draw (3) to node [inner sep=1pt, swap] {\scriptsize $c$} (1);
		\draw (2) to node [inner sep=1pt, swap] {\scriptsize $a$} (5);
		\draw [in=110, out=70, loop] (7) to node [above] {\scriptsize $a,\!b$} ();
		\draw [in=110, out=70, loop] (5) to node [above] {\scriptsize $a,\!b$} ();
		\draw [in=110, out=70, loop] (2) to node [above] {\scriptsize $a,\!c$} ();
		\draw [in=110, out=70, loop] (8) to node [above] {\scriptsize $a,\!c$} ();
		\draw [in=-20, out=-60, loop] (1) to node [right] {\scriptsize $a,\!c$} ();
		\draw [in=240, out=200, loop] (6) to node [left] {\scriptsize $a,\!b$} ();
		\draw [in=240, out=200, loop] (4) to node [below] {\scriptsize $b,\!c$} ();
		\draw [in=-20, out=-60, loop] (3) to node [below] {\scriptsize $b,\!c$} ();
		\draw [in=-110, out=-70, loop] (9) to node [below] {\scriptsize $b,\!c$} ();
		\draw [in=110, out=70, loop] (0) to node [above] {\scriptsize $a,\!b,\!c$} ();
	\end{pgfonlayer}
\end{tikzpicture}

\end{center}
\end{example}

\section{Completeness results}
\label{sec:completeness}

In this section, we show the completeness results that we mentioned after \cref{prop:soundness}. Namely, we will see that the axiom system $\SC$ (resp., $\SCmin$, $\SCmax$) is complete with respect to the class of generalized (resp., minimal, maximal) simplicial models.
In the presence of the distributed knowledge operator, completeness proofs usually proceed in two steps (see e.g.\ \cite{BaltagS20,FaginHV92}).
First, we define a \emph{canonical pseudo-model} whose worlds are maximal consistent sets of formulas. Then, this pseudo-model needs to be \emph{unravelled} in order to obtain an actual model.

We follow these two routine steps in \cref{sec:Mc} and \cref{sec:unravel}, where we recall the definitions and main properties of the canonical pseudo-model and the unravelling construction. Even though our setting is slightly non-standard (with partial epistemic frames and extra axioms), everything works as usual in these two sections.
\cref{sec:proprifying} deals with the fact that our models are proper.
Finally in \cref{sec:completeness-proof}, we put all the pieces together to show that $\SC$ is complete with respect to the class of proper partial epistemic models with no empty world.
Completeness for generalized simplicial models then follows directly from \cref{thm:truth-equiv}.
The proofs of completeness for $\SCmin$ and $\SCmax$ work the same, with a couple of extra conditions to be checked at the end. Hence we focus on $\SC$ for the time being.

\subsection{The canonical pseudo-model}
\label{sec:Mc}


A pseudo-model is similar to a Kripke model, except that we have an indistinguishability relation $\sim_B$ for each group of agents $B \subseteq A$.
In the context of this paper, we will consider pseudo-models where those relations $\sim_B$ are partial equivalence relations (PER).
Any partial epistemic model (as in \cref{def:partial-epistemic-model}) yields a pseudo-model by setting $\sim_B\, = \bigcap_{a \in B} \sim_a$. 
However, in general, this equality may not hold in a pseudo-model.

\begin{definition}
A \emph{pseudo-model} $M = \la W, \sim, L \ra$ over the set of agents $A$ consists of:
\begin{itemize}
\item a set of worlds $W$;
\item a PER $\sim_B\,\subseteq W \times W$ for each $B \subseteq A$,
such that (i) $\sim_{B'}\, \subseteq\, \sim_{B}$ whenever $B \subseteq B' \subseteq A$, and (ii) for every $w \in W$ and $B,B'\subseteq A$, if $w \sim_B w$ and $w \sim_{B'} w$, then $w \sim_{B\cup B'} w$;
\item a valuation function $L : W \to \Pow{\AP}$.
\end{itemize}
\end{definition}

The satisfaction relation $M,w \models \phi$ on pseudo-models is defined inductively on the structure of the formula $\phi \in \mathcal{L}_D$, as we did in \cref{sec:partial-epistemic-models}, except that to define the semantics of the distributed knowledge operator $D_B\,\phi$ we rely on the relation $\sim_B$ of the pseudo-model, rather than the intersection of the single-agent relations.

Let $\Gamma \subseteq \mathcal{L}_D$ be a set of formulas. We write $\Gamma \vdash_{\SC} \phi$ when the formula $\phi$ is provable from the hypothesis~$\Gamma$ in the proof system $\SC$. 
We say that $\Gamma$ is \emph{consistent} when $\Gamma \not \vdash_{\SC} \false$, and that $\Gamma$ is \emph{maximal consistent} when moreover, for every $\phi \not \in \Gamma$, we have $\Gamma \cup \{\phi\} \vdash_{\SC} \false$.

\begin{definition}
\label{def:Mc}
The \emph{canonical pseudo-model} $\Mc = \la \Wc, \simc, \Lc \ra$ is defined as follows:
\begin{itemize}
\item $\Wc = \{ \Gamma \mid \Gamma \textup{ is a maximal consistent set of formulas} \}$.
\item $\Gamma \simc_B \Delta$ \; iff \; $D_B\,\phi \in \Gamma$ implies $\phi \in \Delta$.
\item $\Lc(\Gamma) = \Gamma \cap \AP$.
\end{itemize}
\end{definition}

First, let us check that $\Mc$ is indeed a pseudo-model. Symmetry and transitivity of $\simc_B$ are proved as usual using Axioms $\mathbf{B}$ and $\mathbf{4}$, respectively.
To see that $\simc_{B'}\,\subseteq\,\simc_{B}$ for $B \subseteq B'$, assume that $\Gamma \simc_{B'} \Delta$ and that $D_B\,\phi \in \Gamma$.
Using the axiom \textbf{Mono} and the fact that $\Gamma$ is maximal consistent, we must have $D_{B'}\,\phi \in \Gamma$.
Then $\phi \in \Delta$ because we assumed $\Gamma \simc_{B'} \Delta$, so $\Gamma \simc_{B} \Delta$ as required.
Finally, assuming that $\Gamma \simc_{B} \Gamma$ and $\Gamma \simc_{B'} \Gamma$, we want to show that $\Gamma \simc_{B\cup B'} \Gamma$.
First, notice that $\aliveprop{B} \in \Gamma$: otherwise, we would have $D_B\,\false \in \Gamma$, i.e.\ $\false \in \Gamma$, and so $\Gamma$ would be inconsistent.
Similarly, $\aliveprop{B'} \in \Gamma$, so $\aliveprop{B \cup B'} \in \Gamma$ by axiom \textbf{Union}.
Let $\Delta^- = \{ \phi \mid D_{B\cup B'}\,\phi \in \Gamma \}$.
Then $\Delta^-$ is consistent, otherwise we would have $D_{B \cup B'}\, \false \in \Gamma$, which we ruled out.
We can thus extend $\Delta^-$ to a maximal and consistent set $\Delta \supseteq \Delta^-$, which satisfies $\Gamma \simc_{B \cup B'} \Delta$.
By symmetry and transitivity of $\simc_{B \cup B'}$, we get $\Gamma \simc_{B \cup B'} \Gamma$ as required.

\begin{lemma}[Truth Lemma]
\label{lem:truth}
For any formula $\phi \in \mathcal{L}_D$ and any maximal consistent set of formulas $\Gamma \in \Wc$,
we have $\phi \in \Gamma$ iff $\Mc,\Gamma \models \phi$.
\end{lemma}
\begin{proof}
Proceed by induction on the structure of $\phi$.
The base case of atomic propositions holds by definition of $\Lc$. For the boolean connectives, the proof is trivial.

Let us do the case of $D_B\,\phi$.
Assume that $D_B\,\phi \in \Gamma$ and let $\Delta \in \Wc$ such that $\Gamma \simc_{B} \Delta$. 
By definition of $\simc$, we have $\phi \in \Delta$, so by induction hypothesis, $\Mc,\Delta \models \phi$.
Thus $\Mc,\Gamma \models D_B\,\phi$.
Conversely, assume that $\Mc,\Gamma \models D_B\,\phi$, and suppose by contradiction that $D_B\,\phi \not \in \Gamma$.
Then the set $\Delta^- = \{ \neg \phi \} \cup \{ \psi \mid D_B\,\psi \in \Gamma \}$ is consistent.
Indeed, if $\Delta^-$ was inconsistent, we would have a proof of $\vdash_{\SC} \psi_1 \land \ldots \land \psi_k \Rightarrow \phi$ where $D_B\,\psi_i \in \Gamma$ for every~$i$. 
Then, using Axiom~\textbf{K}, we could prove $\vdash_{\SC} D_B\,\psi_1 \land \ldots \land D_B\,\psi_k \Rightarrow D_B\,\phi$. Because $\Gamma$ is maximal consistent, this implies that $D_B\,\phi \in \Gamma$, which contradicts our assumption.
So $\Delta^-$ is consistent, and by Lindenbaum's Lemma, we can extend it to a maximal consistent set $\Delta \supseteq \Delta^-$.
By construction, $\Gamma \simc_{B} \Delta$, and by induction hypothesis, $\Mc, \Delta \not \models \phi$.
This contradicts the initial assumption that $\Mc,\Gamma \models D_B\,\phi$. Therefore $D_B\,\phi \in \Gamma$, which concludes the proof.
\end{proof}

\begin{remark}
	In this article, pseudo-models serve only as a means to show the completeness of $\SC$.
	It is possible, however, to take pseudo-models as a primitive notion, and to define a semantics for $\mathcal{L}_D$ based on them.
	This yields a non-standard notion of distributed knowledge.
	This approach has been studied in a companion paper~\cite{goubaultSemisimplicialSetModels2023}.
	Remarkably, pseudo-models also have a geometric counterpart: they amount to replacing simplicial complexes by semi-simplicial sets.
	Another paper that used such pseudo-models as the main object of study is~\cite{baltagCorrelatedKnowledgeEpistemicLogic2010}, in order to model observability in quantum systems. 
\end{remark}

\subsection{Unravelling a pseudo-model}
\label{sec:unravel}

As we mentioned at the beginning of \cref{sec:Mc}, partial epistemic models can be viewed as a special case of pseudo-models.
However, the canonical model $\Mc$ is not among this subclass of pseudo-models, because $\simc_B \,\neq \bigcap_{a \in B} \simc_{\{a\}}$.
We now describe a general construction called \emph{unravelling}, which can turn any pseudo-model into a (bisimilar) partial epistemic model.
Later, we will use this construction to unravel the canonical model.

Let $M = \la W, \sim, L \ra$ be a pseudo-model.
A \emph{history} of $M$ is a finite sequence of the form $h = (w_0, B_1, w_1, \ldots, B_k, w_k)$ for some $k \geq 0$, such that $w_{i-1} \sim_{B_{i}} w_{i}$ for all $1 \leq i \leq k$.
We write $\last(h) = w_k$ for the last element of a history, and we write $h \rightarrow_a h'$ if $h' = (h, B_{k+1}, w_{k+1})$ with $a \in B_{k+1}$

\begin{definition}
The \emph{unravelling} of $M$ is a partial epistemic model $\Mu = \la \Wu, \simu, \Lu \ra$ defined as follows:
\begin{itemize}
\item $\Wu$ is the set of histories of $M$,
\item $\simu_a$ is the transitive and symmetric closure of $\rightarrow_a$, i.e., $\simu_a\, = \left( \rightarrow_a \cup \leftarrow_a \right)^+$,
\item $\Lu(h) = L(\last(h))$.
\end{itemize}
\end{definition}

It is immediate to see that $U(M)$ is a partial epistemic model, since $\simu_a$ is symmetric and transitive by definition.
Before we can prove that unravelling a pseudo-model preserves the satisfaction relation (\cref{lem:unravelling-preserves-truth}), we first show a useful lemma relating the relation $\sim_B$ of a pseudo-model with the one of its unravelling.

\begin{lemma}
\label{lem:unravelling-preserves-simB}
Let $M$ be a pseudo-model and $U(M)$ its unravelling.
Let $h, h' \in H$ be histories, and $B \subseteq A$ a set of agents.
If $h \simu_a h'$ for all $a \in B$, then $\last(h) \sim_B \last(h')$.
\end{lemma}
\begin{proof}
Let us first assume that $h \neq h'$; we will treat the other case separately.
Let~$h''$ be the common prefix of $h$ and $h'$, and let us write $h = (h'', B_1, w_1, \ldots, B_k, w_k)$ and $h' = (h'', B'_1, w'_1, \ldots, B'_\ell, w'_\ell)$.
For each agent $a \in B$, notice that there is a unique non-redundant path from $h$ to $h'$ for the relation $(\rightarrow_a \cup \leftarrow_a)$, which first goes backwards from $h$ to $h''$, then forwards from $h''$ to $h'$, as follows:
 $h \leftarrow_a \ldots \leftarrow_a h'' \rightarrow_a \ldots \rightarrow_a h'$.
Since any proof that $h \simu_a h'$ must go through this path, we must have $a \in B_i$ for all $1 \leq i \leq k$, and $a \in B'_j$ for all $1 \leq j \leq \ell$.
The same fact holds for each $a \in B$, so in fact $B \subseteq B_i$ and $B \subseteq B'_j$ for all $i,j$, 
and since $M$ is a pseudo-model, $\sim_{B_i} \,\subseteq\, \sim_B$ and $\sim_{B'_j} \,\subseteq\, \sim_B$.
Thus, all the worlds of~$M$ along this path are related by $\sim_B$:
\[\last(h) = w_k \sim_B \ldots \sim_B w_1 \sim_B \last(h'') \sim_B w'_1 \sim_B \ldots \sim_B w'_\ell = \last(h').\]
Finally, by transitivity of $\sim_B$, we get $\last(h) \sim_B \last(h')$ as required.

We still need to prove the lemma for $h = h'$.
The difficulty is that to have $h \simu_a h$,
we must take a detour via another history $h \rightarrow_a h'' \leftarrow_a h$.
However, unlike in the first half of the proof, the choice of $h''$ might differ for each $a \in B$.
This is where condition (ii) in the definition of a pseudo-model comes into play.
Clearly, for each $a \in B$, $h \simu_a h$ implies that $\last(h) \sim_{\{a\}} \last(h)$.
Using condition (ii) of the pseudo-model $M$ repeatedly, we get $\last(h) \sim_{B} \last(h)$, which concludes the proof.
\end{proof}

\begin{lemma}
\label{lem:unravelling-preserves-truth}
For every history $h \in H$ and formula $\phi \in \mathcal{L}_D$, $M,\last(h) \models \phi$ iff $U(M),h \models \phi$.
\end{lemma}
\begin{proof}
This is proved by induction on the structure of the formula $\phi$. The cases of atomic propositions and boolean connectives are straightforward, so we focus on the case of $D_B\,\phi$.

For the left-to-right implication, assume that $M,\last(h) \models D_B\,\phi$, and let $h' \in H$ be a history such that $h \simu_B h'$, i.e., $h \simu_a h'$ for all $a \in B$.
By \cref{lem:unravelling-preserves-simB} we get $\last(h) \sim_B \last(h')$, which implies that $M,\last(h') \models \phi$, and by induction hypothesis $U(M),h' \models \phi$.

For the right-to-left implication, assume that $U(M),h \models D_B\,\phi$ and let $w' \in W$ such that $\last(h) \sim_B w'$ in $M$. Consider the history $h' = (h, B, w')$.
Then $h \rightarrow_a h'$ for each $a \in B$, therefore, $h \simu_B h'$.
Thus $U(M),h' \models \phi$ because we assumed that $U(M),h \models D_B\,\phi$, and by induction hypothesis, $M,\last(h') \models \phi$ i.e.\ $M,w' \models \phi$ as required.
\end{proof}

\begin{remark}
In fact, the map $\last : H \to W$ can be shown to be a bisimulation between $M$ and $U(M)$.
\end{remark}

\subsection{Making the model proper}
\label{sec:proprifying}

Even though the canonical model $\Mc$ can be shown to be proper thanks to Axiom~\textbf{P}, the unravelling construction introduces some redundancy and as a consequence, $U(\Mc)$ is not proper.
However, as we will see in the next section, $U(\Mc)$ has a good enough property: two ``equivalent'' worlds always satisfy the same set of formulas.
This allows us to construct a bisimilar proper model, by removing the redundant worlds, as we describe in this section.

Let $M = \la W, \sim, L \ra$ be a partial epistemic model, and recall that $\live{w} = \{ a \in A \mid w \sim_a w \}$ is the set of alive agents in~$w$.
We say that two worlds $w,w'$ are \emph{equivalent}, written $w \equiv w'$, if $\live{w} = \live{w'}$ and for all $a \in \live{w}$, $w \sim_a w'$.
Thus, $M$ is proper if and only if $w \equiv w'$ implies $w = w'$.
Here, we assume a weaker property: that if $w \equiv w'$, then $L(w) = L(w')$.
From this, one can deduce by an easy induction that for all $\phi \in \mathcal{L}_D$, $M,w \models \phi$ iff $M,w' \models \phi$.

\begin{definition}
The model $\Mproper = \la \Wproper, \simproper, \Lproper \ra$ is defined as follows:
\begin{itemize}
\item $\Wproper$ is the set of equivalence classes of the relation $\equiv$. We write $[w] \in \Wproper$ for the equivalence class of $w \in W$.
\item $[w_1] \simproper_a [w_2]$ iff $w_1 \sim_a w_2$.
\item $L'([w]) = L(w)$.
\end{itemize}
\end{definition}

It is straightforward to check that the definitions of $\simproper$ and $\Lproper$ do not depend on the choice of representative of the equivalence class, and that $\Mproper$ is a partial epistemic model.

\begin{lemma}
\label{lem:proper-preserves-truth}
The model $\Mproper$ is proper, and moreover $M,w \models \phi$ iff $\Mproper,[w] \models \phi$.
\end{lemma}
\begin{proof}
To see that $\Mproper$ is proper, notice that $[w] \equiv [w']$ implies that $w \equiv w'$. So~$w$ and~$w'$ belong to the same equivalence class, i.e.\ $[w] = [w']$ and the model is proper.

To prove the second part of the lemma, proceed by induction on the formula~$\phi$.
The cases of atomic propositions and boolean connectives are trivial.
So assume that $M,w \models D_B\,\phi$, and let $[w'] \in \Wproper$ be such that $[w] \simproper_a [w']$ for all $a \in B$.
Then $w \sim_a w'$ for all $a \in B$, thus $M,w' \models \phi$. By induction hypothesis, $\Mproper,[w'] \models \phi$.
The converse is identical.
%
\end{proof}

\subsection{Proofs of completeness}
\label{sec:completeness-proof}

We are almost ready to prove completeness for the axiom system $\SC$.
What remains to be checked is that the unravelled canonical model $U(\Mc)$ can be made proper using the construction in \cref{sec:proprifying}, and that the resulting model $\proper{U(\Mc)}$ has no empty world.

\begin{lemma}
\label{lem:alive-history}
Let $h$ be a history of $\Mc$, and write $\Gamma = \last(h)$.
Then $a \in \live{h}$ iff $\aliveprop{a} \in \Gamma$.
\end{lemma}
\begin{proof}
If $a \in \live{h}$, there must be some history $h_0$ such that $h \rightarrow_a h_0$. Writing $\Gamma' = \last(h_0)$, this means that $\Gamma \simc_a \Gamma'$.
So we cannot have $K_a\,\false \in \Gamma$, otherwise $\false \in \Gamma'$ and $\Gamma'$ would be inconsistent.
Thus $\neg K_a\,\false \in \Gamma$ (because $\Gamma$ is maximal), i.e.\ $\aliveprop{a} \in \Gamma$.

Conversely, assume that $\aliveprop{a} \in \Gamma$.
We have seen in \cref{sec:Axiomatization} that the formula $\aliveprop{a} \Rightarrow (K_a\,\phi \Rightarrow \phi)$ is valid in $\KBfour$ (a fortiori in~$\SC$).
Thus for every formula $\phi$, $K_a\,\phi \in \Gamma$ implies $\phi \in \Gamma$, i.e.\ $\Gamma \simc_a \Gamma$.
Writing $h' = (h, \{a\}, \Gamma)$, we have shown that $h \rightarrow_a h'$.
By symmetry and transitivity, this yields $h \simu_a h$, which concludes the proof.
\end{proof}

\noindent
The following condition is required in order to apply the construction of \cref{sec:proprifying}.

\begin{proposition}
In the model $U(\Mc)$, if $h \equiv h'$ then $\Lu(h) = \Lu(h')$.
\end{proposition}
\begin{proof}
Consider two histories $h, h'$ of $\Mc$, and assume that $h \equiv h'$, i.e.\ $\live{h} = \live{h'}$ and for all $a \in \live{h}$, $h \simu_a h'$.
Let $\Gamma = \last(h)$ and $\Delta = \last(h')$.
By \cref{lem:alive-history}, $\aliveprop{a} \in \Gamma \iff \aliveprop{a} \in \Delta$.

Let $p \in \Lu(h) = \Lc(\Gamma)$ an atomic proposition. We have $p \in \Gamma$ by definition of $\Lc$.
Let $B = \live{h} =  \{ a \mid \aliveprop{a} \in \Gamma \}$.
Since $\Gamma$ is maximal and consistent, it contains the formula $\aliveprop{B} \land \deadprop{\complement{B}} \land p$.
By Axiom~\textbf{P}, $\Gamma$ must also contain $D_B\,(\deadprop{\complement{B}} \Rightarrow p)$.
By \cref{lem:unravelling-preserves-simB}, $h \simu_a h'$ for all $a \in B$ implies that $\Gamma \simc_B \Delta$.
By definition of~$\simc$, the set $\Delta$ then contains the formula $\deadprop{\complement{B}} \Rightarrow p$.
And since $\Delta$ is maximal consistent, and contains the formula $\deadprop{\complement{B}}$, we finally have $p \in \Delta$, i.e.\ $p \in \Lu(h')$.

The converse inclusion $\Lu(h') \subseteq \Lu(h)$ is proved symmetrically.
\end{proof}

\begin{proposition}
\label{prop:Mc-proper}
The model $\proper{U(\Mc)}$ has no empty world.
\end{proposition}
\begin{proof}
It is sufficient to show that $U(\Mc)$ has no empty world, since any agent $a$ which is alive in $h$ is also alive in $[h]$, because $h \simu_a h$ implies $[h] \simproper_a [h]$.

So let $h$ be a history of $\Mc$, and write $\Gamma = \last(h)$. We want to find some agent $a \in A$ such that $a \in \live{h}$.
Since $\Gamma$ is maximal and consistent, and obeys the Axiom~\textbf{NE}, there must be some agent $a \in A$ such that $\aliveprop{a} \in \Gamma$.
By \cref{lem:alive-history}, this entails $a \in \live{h}$.
\end{proof}

\begin{theorem}
\label{thm:completeness-SC-kripke}
The system $\SC$ is complete with respect to the class of proper partial epistemic models with no empty world.
\end{theorem}
\begin{proof}
We prove the converse of completeness: if a formula $\phi \in \mathcal{L}_D$ is not provable, then it is not valid in all models.
So assume that $\not \vdash_{\SC} \phi$, i.e.\ $\{\neg \phi\}$ is a consistent set of formulas.
By Lindenbaum's Lemma, there is a maximal consistent set $\Gamma$ such that $\neg \phi \in \Gamma$.
By the Truth Lemma (\cref{lem:truth}), $\Mc,\Gamma \models \neg \phi$, and by \cref{lem:unravelling-preserves-truth,lem:proper-preserves-truth}, $\proper{U(\Mc)}, [(\Gamma)] \models \neg \phi$.
Since $\proper{U(\Mc)}$ is a proper partial epistemic model with no empty world, this concludes the proof.
\end{proof}

While \cref{thm:completeness-SC-kripke} might seem somewhat arbitrary, our real goal was to prove completeness with respect to the class of generalized simplicial models:

\begin{corollary}
\label{thm:completeness-SC-simplicial}
The system $\SC$ is complete with respect to the class of simplicial models.
\end{corollary}
\begin{proof}
Assume a formula $\phi \in \mathcal{L}_D$ is valid in all simplicial models.
By \cref{thm:truth-equiv}, $\phi$ is also valid in all proper partial epistemic models with no empty worlds.
So by \cref{thm:completeness-SC-kripke}, $\phi$ is provable in the system $\SC$.
\end{proof}

\paragraph*{Completeness for $\SCmin$ and $\SCmax$.} We now prove completeness of $\SCmin$ and $\SCmax$ with respect to the class of minimal (resp.\ maximal) simplicial models.
The proof is almost the same as the one for $\SC$: we write $\Mcmin$ and $\Mcmax$ for the canonical pseudo-models whose worlds are sets of formulas that are maximal and consistent with respect to the logic $\SCmin$ (resp.\ $\SCmax$).
All the machinery of \cref{sec:Mc,sec:unravel,sec:proprifying} works the same.
The only extra properties that we need to show are the following:

\begin{proposition}
The partial epistemic model $\proper{U(\Mcmin)}$ is minimal, and the partial epistemic model $\proper{U(\Mcmax)}$ is maximal.
\end{proposition}
\begin{proof}
To prove that $\proper{U(\Mcmin)}$ is minimal, it is sufficient to show that $U(\Mcmin)$ is minimal.
Let $h,h'$ be histories of $\Mcmin$ such that $\live{h} \subsetneq \live{h'}$, and let us write $B = \live{h}$, $\Gamma = \last(h)$ and $\Delta = \last(h')$.
Assume for contradiction that for all $a \in B$, $h \simu_a h'$. By \cref{lem:unravelling-preserves-simB}, this entails ${\Gamma \simc_B \Delta}$.
Since $\Gamma$ is maximal and consistent, and using \cref{lem:alive-history}, the formula $\aliveprop{B} \land \deadprop{\complement{B}}$ belongs to $\Gamma$. Using Axiom~\textbf{Min}, $\Gamma$ must also contain the formula $D_B\,\deadprop{\complement{B}}$, and since $\Gamma \simc_B \Delta$, we obtain that $\deadprop{\complement{B}} \in \Delta$.
But this is a contradiction: since we assumed that $\live{h} \subsetneq \live{h'}$, there exists an agent $a \not \in B$ such that $a \in \live{h'}$, i.e.\ $\aliveprop{a} \in \Delta$ by \cref{lem:alive-history}.

For the second part of the statement, again it suffices to prove that $U(\Mcmax)$ is maximal.
Let $h'$ be a history of $\Mcmax$, with $\Delta = \last(h')$, and let $B \subsetneq \live{h'}$.
We want to exhibit a sub-world~$h$ of~$h'$ whose set of alive agents is~$B$.
For every $a \in B$, we have $\aliveprop{a} \in \Delta$ by \cref{lem:alive-history}, so using Axiom~\textbf{Max} and the fact that $\Delta$ is maximal and consistent, we get $\neg D_B \neg \deadprop{\complement{B}} \in \Delta$.
Then the set $\Gamma^- = \{ \deadprop{\complement{B}} \} \cup \{ \psi \mid D_B\,\psi \in \Delta \}$ is consistent, using the same reasoning as in the proof of \cref{lem:truth}.
By Lindenbaum's Lemma, there is a maximal consistent set $\Gamma \supseteq \Gamma^-$.
Moreover, $\Gamma \simc_B \Delta$ by construction (and symmetry of~$\simc$). 
Let $h = (h', B, \Gamma)$. Then we have $h \simu_a h'$ for every $a \in B$, so in particular $B \subseteq \live{h}$. 
The converse inclusion stems from the fact that $\deadprop{\complement{B}} \in \Gamma$ and \cref{lem:alive-history}.
Hence $h$ is a sub-world of $h'$ such that $\live{h} = B$.
\end{proof}

With the above proposition, and using the same reasoning as before, we get a proof of completeness of $\SCmin$ and $\SCmax$ with respect to the classes of minimal/maximal proper partial epistemic models with no empty world.
More interestingly, we can lift this to simplicial models, once again using \cref{thm:truth-equiv}, and the fact that the notions of minimal/maximal models are preserved by the equivalence (see \cref{prop:sigma-well-defined}).
Finally:

\begin{theorem}
\label{thm:completeness-SC-min-max}
The proof system $\SCmin$ (resp.\ $\SCmax$) is complete with respect to the class of minimal (resp.\ maximal) simplicial models.
\end{theorem}

\section{Dynamics: communication pattern models}

\label{sec:dynamics}

In this section, we describe how a generalized simplicial model evolves when the agents share information by communicating.
We use the framework of \emph{communication patterns}~\cite{BaltagS20,Castaneda21pattern},
which we slightly modify in two ways: (i) we define it entirely on (generalized) simplicial models, rather than Kripke models, and (ii) we allow the processes to crash during a communication event.
The first modification was also performed in~\cite{Castaneda22pattern} (Definition 24), in a setting without crashes.
Their proposed definition is very similar to our \cref{def:pattern-update}; in fact, it is a special case of it.
The second point, adding the possibility of crashes, has not been done previously with communication patterns to our knowledge.
Conceptually this is quite straightforward, but some care is required in order to avoid some technical issues (see \cref{rem:ell-well-defined} and \cref{ex:broadcast-comm-pattern}).
Similar issues arise when we add the possibility of crashes to the action model formalism, as noticed in~\cite{nakai2023partial}.

\paragraph*{Local simplicial models}

In this section, in contrast to the rest of the paper, we will adopt the distributed computing practice of labelling the vertices (rather than the worlds) of a simplicial model with atomic propositions.
Thus, as in previous papers (e.g.~\cite{gandalf-journal, Ditmarsch2020KnowledgeAS, Ditmarsch21}), we assume that the set $\AP$ of atomic proposition is partitioned into sets $\AP = \bigcup_{a \in A} \AP_a$, so that each atomic proposition ``belongs'' to a particular agent.
Then a \emph{local simplicial model} $\C = \la V, S, \chi, W, \ell \ra$ is given by a chromatic simplicial complex $\la V, S, \chi \ra$ and a distinguished set of worlds~$W$, as in \cref{def:generalized-simplicial-model}, except that the labelling $\ell$ assigns to each vertex $v \in V$ of color~$\chi(v) = a$, a set of atomic propositions concerning agent~$a$, $\ell(v) \subseteq \AP_a$.

Note that every local simplicial model gives rise to a (generalized) simplicial model in the sense of \cref{def:generalized-simplicial-model}: the labelling of a given world~$w \in W$ is then obtained by taking the union of the labellings of its vertices: $\ell(w) = \bigcup_{v \in w} \ell(v)$.
Local simplicial models are strictly less general than the  simplicial models of \cref{def:generalized-simplicial-model}.
Indeed, local simplicial models obey the so-called \emph{Locality axiom} (see~\cite{gandalf-journal}), which says that every agent~$a$ knows the status (true or false) of all the atomic propositions in $\AP_a$.
The locality assumption will be crucial when we define the product update model (see \cref{rem:ell-well-defined}).

\paragraph*{Communication patterns}


Communication patterns 
rely on communication graphs, which indicate how information flows between the agents: an arrow from~$a$ to~$b$ in a communication graph indicates that agent~$a$ successfully sends a message to agent~$b$, containing all the information currently known to~$a$.
In distributed computing, this is known as a \emph{full-information protocol}.
In~\cite{Castaneda22pattern}, communication graphs are always assumed to be reflexive, so that each agent remembers the information that they had at the previous round. Here, we relax this assumption, and inspired by \cref{sec:partial-epistemic-models} we will interpret lack of reflexivity as representing the death (a.k.a.\ crash, in distributed computing) of an agent.

\begin{definition}[Communication pattern]
A \emph{communication graph} $G \subseteq A \times A$ is a binary relation on the set of agents. When $G$ is clear from context, we write $a \to b$ instead of $(a,b) \in G$. The \emph{in-neighbourhood} of~$a$ in~$G$ is denoted $\inN{G}{a} = \{ b \in A \mid b \to a \}$, and the \emph{out-neighbourhood} is $\outN{G}{a} = \{ b \in A \mid a \to b \}$.
We say that agent~$a$ is \emph{alive} in~$G$ when $a \to a$, and that~$a$ is \emph{dead} otherwise.
A~\emph{communication pattern} $P$ is a set of communication graphs, i.e.\ $P \subseteq \mathscr{P}(A \times A)$.
\end{definition}

Communication patterns describe a round-based communicative event where every agent tries to broadcast its current local state to all other agents; but some of those messages might be lost.
At each round, a communication graph~$G \in P$ is chosen arbitrarily, and describes which messages failed to arrive during this round: an edge $a \to b$ in~$G$ indicates that~$a$ successfully delivered its message to~$b$. 
Moreover, some agents might crash during the round, possibly after sending messages to other agents. A crash is indicated by the lack of a reflexive edge $a \to a$ in~$G$.
Communication patterns are closely related to \emph{dynamic networks}~\cite{Kuhn11dynamic}, a very general distributed computing model which subsumes not only message-passing models but also round-based shared memory models such as immediate snapshot.

Given a (local) simplicial model $\C = \la V, S, \chi, W, \ell \ra$ and a communication pattern~$P$, we denote by $\C \odot P$ the \emph{updated simplicial model} which represents the knowledge of the agents after some communicative event $G \in P$ occurred.
Informally, its worlds should be pairs $(w,G)$ where $w \in W$ is a world of~$\C$ and $G \in P$ is a communication graph allowed by~$P$. Moreover, we require that~$G$ is \emph{compatible} with~$w$, in the sense that agents that are dead in~$w$ cannot send messages in~$G$: $a \not\in \chi(w)$ implies $\outN{G}{a} = \emptyset$.
Two worlds $(w,G)$ $(w', G')$ should be indistinguishable by some agent~$a$ when in both communication graphs~$G$ and~$G'$, $a$ has received messages from the same set of agents, and the worlds $w$ and $w'$ are indistinguishable for all of these agents, i.e.\ $(w,G) \sim_a (w', G')$ iff $\inN{G}{a} = \inN{G'}{a}$ and $\inN{G}{a} \subseteq \chi(w \cap w')$ and $a \in \inN{G}{a}$. Note that the last condition ensures that~$a$ is alive in $w, w', G$ and~$G'$.
One could check that this yields a partial epistemic model\footnote{This partial epistemic model might not be proper in general. One can make it proper as in \cref{sec:proprifying}, but this requires an extra assumption. \cref{rem:ell-well-defined} discusses the same issue in the simplicial setting.}; but in the definition below, we directly construct the corresponding simplicial model.

Let us first introduce some notations.
Given a vertex~$v$ of a simplicial model and a set~$B$ of agents, we write $\star{B}{v}$ for the set of simplexes coloured by~$B$ containing~$v$.
\[ \star{B}{v} = \{ X \in S \mid \chi(X) = B \text{ and } v \in X \}\]
Given a world~$w$ of a simplicial model and a set~$B \subseteq \chi(w)$, we write $\restrict{w}{B}$ for the sub-simplex of~$w$ containing exactly the vertices whose colour is in~$B$. Note that $\restrict{w}{B}$ need not be a world in general.
We will use the simplex $\restrict{w}{B}$ to represent the new local state of an agent~$a$ after it receives (full-information) messages from the set~$B$ of agents.
Finally, to increase readability, we annotate vertices with their colour, e.g.\ we write $v_a \in V$ as shorthand for $v \in V$ such that $\chi(v) = a$. For instance, $\restrict{w}{B} = \{ v_a \in w \mid a \in B \}$.

\begin{definition}
\label{def:pattern-update}
The \emph{updated simplicial model} is given by $\C \odot P = \la V', S', \chi', W', \ell' \ra$, where:
\begin{itemize}
\item $V' = \{ (v_a, X) \mid v_a \in V \text{ and } X \in \star{B}{v_a} \text{ where } B = \inN{G}{a} \text{ for some } G \in P \}$.
\item $S' = \dclosure{W'}$, the downward-closure of $W'$.
\item $\chi'(v_a,X) = a$.
\item $W' = \{ w \odot G \mid w \in W \text{, } G \in P \text{ and } G \text{ is compatible with } w \}$,\\
where $w \odot G = \{ (v_a, \restrict{w}{\inN{G}{a}}) \mid v_a \in w \text{ and $a$ is alive in $G$} \}$.
\item $\ell'(v_a, X) = \ell(v_a)$.
\end{itemize}
\end{definition}

\begin{remark} \label{rem:ell-well-defined}
Note that it is possible to have $w \odot G = w' \odot G'$ for two distinct worlds~$w,w'$ and communication graphs $G, G'$.
This is due to the possibility of crashing agents, as illustrated \cref{ex:broadcast-comm-pattern} below.
This is where the requirement that the initial simplicial model~$\C$ must be local becomes crucial.
Indeed, in a non-local model, we put atomic propositions on the worlds, not vertices, so the last item of the definition should become $\ell'(w \odot G) = \ell(w)$.
However, this is not well-defined when $w \odot G = w' \odot G'$ and $\ell(w) \neq \ell(w')$.
Intuitively, two worlds $w$ and $w'$ of the original model have been ``merged'' and we do not know which one to take the labelling from.
Locality ensures that whenever two worlds are merged, they already had the same labelling in the initial model.
\end{remark}

\begin{example}[Synchronous broadcast with crash failures] \label{ex:broadcast-comm-pattern}
We now define the communication pattern that produces the simplicial model of \cref{ex:message-passing}.
Consider the set of agents $A = \{a, b, c\}$ and the following communication graphs on~$A$:
\begin{mathpar}
	\begin{tikzpicture}[auto, scale=1.2, font={\small},
	-{stealth[length=1mm,width=1mm]}, shorten <=1pt,
	{every loop/.style}={looseness=2, min distance=3mm}]
	\node (a) at (0,1) {$a$};
	\node (b) at (-0.5,0) {$b$};
	\node (c) at (0.5,0) {$c$};
	\path (a) edge[in=110,out=70,loop] (a)
		  (b) edge[in=160,out=200,loop] (b)
		  (c) edge[in=20,out=-20,loop] (c)
		  (a.260) edge (b.60)
		  (b.80) edge (a.240)
		  (b.10) edge (c.170)
		  (c.190) edge (b.-10)
		  (a.280) edge (c.120)
		  (c.100) edge (a.300);
	\end{tikzpicture}
\and
	\begin{tikzpicture}[auto, scale=1.2, font={\small},
	-{stealth[length=1mm,width=1mm]}, shorten <=1pt,
	{every loop/.style}={looseness=2, min distance=3mm}]
	\node (a) at (0,1) {$a$};
	\node (b) at (-0.5,0) {$b$};
	\node (c) at (0.5,0) {$c$};
	\path (b) edge[in=160,out=200,loop] (b)
		  (c) edge[in=20,out=-20,loop] (c)
		  (b.80) edge (a.240)
		  (b.10) edge (c.170)
		  (c.190) edge (b.-10)
		  (c.100) edge (a.300);
	\end{tikzpicture}
\and
	\begin{tikzpicture}[auto, scale=1.2, font={\small},
	-{stealth[length=1mm,width=1mm]}, shorten <=1pt,
	{every loop/.style}={looseness=2, min distance=3mm}]
	\node (a) at (0,1) {$a$};
	\node (b) at (-0.5,0) {$b$};
	\node (c) at (0.5,0) {$c$};
	\path (b) edge[in=160,out=200,loop] (b)
		  (c) edge[in=20,out=-20,loop] (c)
		  (a.260) edge (b.60)
		  (b.80) edge (a.240)
		  (b.10) edge (c.170)
		  (c.190) edge (b.-10)
		  (c.100) edge (a.300);
	\end{tikzpicture}
\and
	\begin{tikzpicture}[auto, scale=1.2, font={\small},
	-{stealth[length=1mm,width=1mm]}, shorten <=1pt,
	{every loop/.style}={looseness=2, min distance=3mm}]
	\node (a) at (0,1) {$a$};
	\node (b) at (-0.5,0) {$b$};
	\node (c) at (0.5,0) {$c$};
	\path (b) edge[in=160,out=200,loop] (b)
		  (c) edge[in=20,out=-20,loop] (c)
		  (a.260) edge (b.60)
		  (b.80) edge (a.240)
		  (b.10) edge (c.170)
		  (c.190) edge (b.-10)
		  (a.280) edge (c.120)
		  (c.100) edge (a.300);
	\end{tikzpicture}
\end{mathpar}
We name these graphs $G_1, \ldots, G_4$, from left to right.
Note that we omitted some graphs that can be obtained from those by permuting the names of the agents (i.e., graphs where agent~$b$ or~$c$ crashed instead of~$a$).
Intuitively,
\begin{itemize}
\item $G_1$ is an execution where no crash occurred, all messages were successfully delivered;
\item $G_2$, $G_3$, $G_4$ are executions where only process~$a$ crashed, after sending $0$, $1$ or $2$ messages.
\end{itemize}
Among those communication graphs, only $G_1$, $G_2$ and $G_3$ have ``detectable crashes'', in the sense that whenever a process is dead, at least one of the remaining agents knows about it (because no message was received from the dead agent).
So let us define two communication patterns: $P_{\text{detectable}}$ contains $G_1,G_2,G_3$ as well as graphs obtained from them by permuting the names of the agents (totalling $10$ graphs);
and $P_{\text{undetectable}}$ contains $G_1,G_2,G_3,G_4$ as well as permutations of them (totalling $13$ graphs).

Let $\C$ be the simplicial model which consists of only one triangle world~$w$ with agents~$a,b,c$.
One can check that computing $\C \odot P_{\text{detectable}}$ yields the (minimal) simplicial model of \cref{ex:message-passing}.
Indeed, world $w \odot G_1$ corresponds to the facet $w_0$; world $w \odot G_2$ corresponds to $w_5$; and world $w \odot G_3$ corresponds to~$w_4$.
Similarly, one can check that $\C \odot P_{\text{undetectable}}$ yields three extra worlds, corresponding to the three edges of world $w_0$ in \cref{ex:message-passing}, where one agent has crashed but none of the others know about it.

A more interesting example is to consider what happens when the initial model $\C$ has more than one facet.
In the picture below, we start from the model $\C'$ which comprises two triangle worlds~$w$ and~$w'$ that are glued along their $bc$-coloured edge.
Computing $\C' \odot P_{\text{detectable}}$ gives rise to the simplicial complex depicted on the right, with $19$ worlds named $w_0, \ldots, w_9$ and $w'_0, \ldots, w'_9$ (notice that $w'_5$ is missing).
Similarly, starting with the binary input sphere would yield the same picture as in \cref{fig-synchEvol}.
\begin{center}
	\begin{tikzpicture}[auto,scale=1.2,baseline={(0,0)},
	cloudgrey/.style={draw=black,thick,circle,fill=cyan,inner sep=1pt,minimum size=11pt},cloud/.style={draw=black,thick,circle,fill=magenta!60,inner sep=1pt,minimum size=11pt}, cloudblack/.style={draw=black,thick,circle,fill=green,inner sep=1pt,minimum size=11pt}]
	\draw[thick, draw=black, fill=blue!60, fill opacity=0.15]
	  (4,0) -- (5,-0.577) -- (5,0.577) -- cycle;
	\draw[thick, draw=black, fill=blue!60, fill opacity=0.15]
	  (6,0) -- (5,-0.577) -- (5,0.577) -- cycle;
	\node[cloudgrey] (b1) at (4,0) {$a$};
	\node[cloudgrey] (b2) at (6,0) {$a$};
	\node[cloudblack] (g1) at (5,-0.577) {$c$};
	\node[cloud] (w1) at (5,0.577) {$b$};
	\node at (4.65,0) {$w$};
	\node at (5.35,0.05) {$w'$};
	\node at (5,-1.85) {$\C'$};
	\end{tikzpicture}
	\hspace{2cm}
	\begin{tikzpicture}[scale=1.6, auto, font={\small}, rotate=90]
	\begin{pgfonlayer}{nodelayer}
		\node [style=cyan vertex] (0) at (0, 1) {$a$};
		\node [style=green vertex] (1) at (1, -0.5) {$c$};
		\node [style=magenta vertex] (2) at (-1, -0.5) {$b$};
		\node [style=green vertex] (3) at (1.5, 2) {$c$};
		\node [style=cyan vertex] (4) at (2.5, 0.5) {$a$};
		\node [style=magenta vertex] (5) at (1, -2.25) {$b$};
		\node [style=green vertex] (6) at (-1, -2.25) {$c$};
		\node [style=magenta vertex] (7) at (-1.5, 2) {$b$};
		\node [style=cyan vertex] (8) at (-2.5, 0.5) {$a$};
		\node [style=none] (12) at (-1, -0.5) {};
		\node [style=none] (13) at (0, 1) {};
		\node [style=none] (14) at (1, -0.5) {};
		\node [style=none] (15) at (0, 0) {$w_0$};
		\node [style=cyan vertex] (16) at (0, -5.5) {$a$};
		\node [style=green vertex] (17) at (1, -4) {$c$};
		\node [style=magenta vertex] (18) at (-1, -4) {$b$};
		\node [style=green vertex] (19) at (1.5, -6.5) {$c$};
		\node [style=cyan vertex] (20) at (2.5, -5) {$a$};
		\node [style=magenta vertex] (21) at (1, -2.25) {$b$};
		\node [style=green vertex] (22) at (-1, -2.25) {$c$};
		\node [style=magenta vertex] (23) at (-1.5, -6.5) {$b$};
		\node [style=cyan vertex] (24) at (-2.5, -5) {$a$};
		\node [style=none] (25) at (-1, -4) {};
		\node [style=none] (26) at (0, -5.5) {};
		\node [style=none] (27) at (1, -4) {};
		\node [style=none] (28) at (0, -4.5) {$w'_0$};
		\node [style=none] (29) at (-3, -2.25) {$\C' \odot P_{\text{detectable}}$};
	\end{pgfonlayer}
	\begin{pgfonlayer}{edgelayer}
		\draw [style=thick] (2) to node [inner sep=1pt] {$w_3$} (8);
		\draw [style=thick] (8) to node [inner sep=1pt] {$w_2$} (7);
		\draw [style=thick] (7) to node [inner sep=1pt] {$w_1$} (0);
		\draw [style=thick] (0) to node [inner sep=1pt] {$w_9$} (3);
		\draw [style=thick] (3) to node [inner sep=1pt] {$w_8$} (4);
		\draw [style=thick] (4) to node [inner sep=1pt] {$w_7$} (1);
		\draw [style=thick] (1) to node [inner sep=2pt] {$w_6$} (5);
		\draw [style=thick] (5) to node [inner sep=2pt] {$w_5$} (6);
		\draw [style=thick] (6) to node [inner sep=2pt] {$w_4$} (2);
		\draw [style=filled path] (13.center)
			 to (14.center)
			 to (12.center)
			 to cycle;
		\draw [style=thick] (18) to node [inner sep=1pt, swap] {$w'_3$} (24);
		\draw [style=thick] (24) to node [inner sep=1pt, swap] {$w'_2$} (23);
		\draw [style=thick] (23) to node [inner sep=1pt, swap] {$w'_1$} (16);
		\draw [style=thick] (16) to node [inner sep=1pt, swap] {$w'_9$} (19);
		\draw [style=thick] (19) to node [inner sep=1pt, swap] {$w'_8$} (20);
		\draw [style=thick] (20) to node [inner sep=1pt, swap] {$w'_7$} (17);
		\draw [style=thick] (17) to node [inner sep=2pt, swap] {$w'_6$} (21);
		\draw [style=thick] (21) to (22);
		\draw [style=thick] (22) to node [inner sep=2pt, swap] {$w'_4$} (18);
		\draw [style=filled path] (26.center)
			 to (27.center)
			 to (25.center)
			 to cycle;
	\end{pgfonlayer}
\end{tikzpicture}

\end{center}
For instance, one can check that the worlds~$w_0$ and~$w'_0$ correspond to $w \odot G_1$ and $w' \odot G_1$, respectively.
Similarly, $w_4$ and~$w'_4$ correspond to $w \odot G_3$ and $w' \odot G_3$ (it is a good exercise to verify that these worlds share the same $c$-coloured vertex).
Most interestingly, the world labelled~$w_5$ corresponds to both $w \odot G_2$ and $w' \odot G_2$ at the same time (cf.\ \cref{rem:ell-well-defined}).
Indeed, when the communication graph~$G_2$ occurs, $a$ has crashed and the two agents~$b$ and~$c$ exchange information. But neither~$b$ nor~$c$ is able to distinguish between the initial worlds~$w$ and~$w'$. So no matter whether we started in~$w$ or~$w'$, the two remaining agents end up with the same local state, i.e., $w \odot G_2 = w' \odot G_2$.
This illustrates the fact that working with simplicial complexes automatically makes the model ``proper''. This is because in simplicial models, worlds are not a first-class entity, they are merely a collection of compatible local states, that is, a simplex.
\end{example}

\begin{example}[Immediate snapshot with initial crash failures]
Similarly, the simplicial model of \cref{ex:immediate-snapshot} can be obtained by computing $\C \odot P_{\text{immediate}}$,
where $\C$ is the simplicial model with a single triangle world for three agents $a, b, c$, and $P_{\text{immediate}}$ contains the following communication graphs and their permutations (totalling $25$ graphs):
\begin{mathpar}
	\begin{tikzpicture}[auto, scale=1.2, font={\small},
	-{stealth[length=1mm,width=1mm]}, shorten <=1pt,
	{every loop/.style}={looseness=2, min distance=3mm}]
	\node (a) at (0,1) {$a$};
	\node (b) at (-0.5,0) {$b$};
	\node (c) at (0.5,0) {$c$};
	\path (a) edge[in=110,out=70,loop] (a)
		  (b) edge[in=160,out=200,loop] (b)
		  (c) edge[in=20,out=-20,loop] (c)
		  (a.260) edge (b.60)
		  (b.80) edge (a.240)
		  (b.10) edge (c.170)
		  (c.190) edge (b.-10)
		  (a.280) edge (c.120)
		  (c.100) edge (a.300);
	\end{tikzpicture}
\and
	\begin{tikzpicture}[auto, scale=1.2, font={\small},
	-{stealth[length=1mm,width=1mm]}, shorten <=1pt,
	{every loop/.style}={looseness=2, min distance=3mm}]
	\node (a) at (0,1) {$a$};
	\node (b) at (-0.5,0) {$b$};
	\node (c) at (0.5,0) {$c$};
	\path (a) edge[in=110,out=70,loop] (a)
		  (b) edge[in=160,out=200,loop] (b)
		  (c) edge[in=20,out=-20,loop] (c)
		  (a.260) edge (b.60)
		  (b.80) edge (a.240)
		  (b.10) edge (c.170)
		  (a.280) edge (c.120);
	\end{tikzpicture}
\and
	\begin{tikzpicture}[auto, scale=1.2, font={\small},
	-{stealth[length=1mm,width=1mm]}, shorten <=1pt,
	{every loop/.style}={looseness=2, min distance=3mm}]
	\node (a) at (0,1) {$a$};
	\node (b) at (-0.5,0) {$b$};
	\node (c) at (0.5,0) {$c$};
	\path (a) edge[in=110,out=70,loop] (a)
		  (b) edge[in=160,out=200,loop] (b)
		  (c) edge[in=20,out=-20,loop] (c)
		  (a.260) edge (b.60)
		  (b.80) edge (a.240)
		  (c.190) edge (b.-10)
		  (c.100) edge (a.300);
	\end{tikzpicture}
\and
	\begin{tikzpicture}[auto, scale=1.2, font={\small},
	-{stealth[length=1mm,width=1mm]}, shorten <=1pt,
	{every loop/.style}={looseness=2, min distance=3mm}]
	\node (a) at (0,1) {$a$};
	\node (b) at (-0.5,0) {$b$};
	\node (c) at (0.5,0) {$c$};
	\path (a) edge[in=110,out=70,loop] (a)
		  (b) edge[in=160,out=200,loop] (b)
		  (c) edge[in=20,out=-20,loop] (c)
		  (a.260) edge (b.60)
		  (b.10) edge (c.170)
		  (a.280) edge (c.120);
	\end{tikzpicture}
\\
	\begin{tikzpicture}[auto, scale=1.2, font={\small},
	-{stealth[length=1mm,width=1mm]}, shorten <=1pt,
	{every loop/.style}={looseness=2, min distance=3mm}]
	\node (a) at (0,1) {$a$};
	\node (b) at (-0.5,0) {$b$};
	\node (c) at (0.5,0) {$c$};
	\path (a) edge[in=110,out=70,loop] (a)
		  (b) edge[in=160,out=200,loop] (b)
		  (a.260) edge (b.60)
		  (b.80) edge (a.240);
	\end{tikzpicture}
\and
	\begin{tikzpicture}[auto, scale=1.2, font={\small},
	-{stealth[length=1mm,width=1mm]}, shorten <=1pt,
	{every loop/.style}={looseness=2, min distance=3mm}]
	\node (a) at (0,1) {$a$};
	\node (b) at (-0.5,0) {$b$};
	\node (c) at (0.5,0) {$c$};
	\path (a) edge[in=110,out=70,loop] (a)
		  (b) edge[in=160,out=200,loop] (b)
		  (a.260) edge (b.60);
	\end{tikzpicture}
\and
	\begin{tikzpicture}[auto, scale=1.2, font={\small},
	-{stealth[length=1mm,width=1mm]}, shorten <=1pt,
	{every loop/.style}={looseness=2, min distance=3mm}]
	\node (a) at (0,1) {$a$};
	\node (b) at (-0.5,0) {$b$};
	\node (c) at (0.5,0) {$c$};
	\path (a) edge[in=110,out=70,loop] (a);
	\end{tikzpicture}
\end{mathpar}
The four types of graphs on top, where all processes are alive but some messages might be lost, correspond to the~$13$ facets of the model in \cref{ex:immediate-snapshot}.
The three bottom graphs are those where some initial crash failure(s) occurred: some agents do not participate in the computation. They correspond to the~$9$ edges and~$3$ vertices on the boundary.
\end{example}


\section{Application to fault-tolerant distributed computing}
\label{sec:application}

The goal of this section is to showcase how the epistemic logic machinery developed in this paper can be used to study concrete distributed computing problems.
More precisely, we study in details the following distributed computing problem: how to prove that consensus cannot be solved in the synchronous broadcast model with one round and one crash failure.
The impossibility result itself is well known, and has been studied extensively in the distributed computing literature, with a very precise analysis of the number of rounds required to solve consensus with various crash assumptions, see e.g.~\cite{DworkM90crash, Unbeatable2014, CastanedaFPRRT23}.
Our focus here is merely to see how to extend the proof technique of \cite{gandalf-journal}, in a setting where processes can crash.

Concurrently with our paper, the same example has been considered in \cite{nakai2023partial}.
There are some slight differences between the two proofs however.
First, they describe the dynamics using the notion of \emph{action models}, extended to take into account crashing processes; while we relied on communication pattern models in \cref{sec:dynamics}.
Secondly, the obstruction formula used in the impossibility proof is different: we use a common knowledge operator, while the proof of \cite{nakai2023partial} uses three nested knowledge operators. This is sufficient for the specific one-round toy example being considered, but does not generalize well to multi-round protocols.
Lastly, the task to be solved itself is slightly different, since we discuss some other variants of the binary task specification.

\subsection{Background on task solvability for fault-tolerant distributed systems}

In this section, we will assume the reader is familiar with topological methods to study task solvability in distributed computing.
Namely, the initial state of the processes can be described by an \emph{input complex} $\cI$.
After communicating, the final states of the processes can be described by a \emph{protocol complex} $\cP_\cI$, whose topological structure depends on the communication primitives being used by the processes.
The task to be solved can also be described by a simplicial complex, called the \emph{task complex} $\cT$.
The central result of distributed computing is the Asynchronous Computability Theorem of Herlihy and Shavit:
\begin{theorem}[\cite{HS99}]
A task is solvable by a given protocol if and only if there exists a simplicial map $\delta : \cP_\cI \to \cT$ (satisfying some extra conditions).
\end{theorem}
Thus, a computational question (solvability of a task) is reduced to a topological question (existence of a simplicial map).
A detailed account of topological methods in distributed computing can be found in~\cite{herlihyetal:2013}.
As we have seen, simplicial complexes can also be viewed as models for epistemic logic.
A full reformulation of task solvability in terms of epistemic logic was developed in~\cite{gandalf-journal}.
We briefly recap below the definitions that we will be using here.

Consider a simplicial model $\cI = \la V_{\cI}, S_{\cI}, \chi_{\cI}, W_{\cI}, \ell_{\cI} \ra$ called the \emph{initial simplicial model}.
Each world of $\cI$, with its labeling $\ell$, represents
a possible initial configuration.
Similarly, we will have $\cO=\la V_{\cO}, S_{\cO}, \chi_{\cO}, W_{\cO}, \ell_{\cO} \ra$ a simplicial model for all possible output values. 
In~\cite{gandalf-journal}, we defined a \emph{task} for $\cI$ using a simplicial action model, since we were interpreting DEL. 
Here we take a more ad-hoc approach and simply encode the relation between inputs and outputs that the task should satisfy.
Hence, $\cT=\la V_{\cT}, S_{\cT}, \chi_{\cT}, W_{\cT}, \ell_{\cT} \ra$ is going to be a sub-complex of $\cI \times \cO$, encoding all the allowed combinations of input vectors and output vectors.

%
%

Since the definition of task solvability relies on the existence of a morphism between simplicial models, we need to define what a morphism is in our setting:

\begin{definition}
Let $A=\la V_{A}, S_{A}, \chi_{A}, W_{A}, \ell_{A} \ra$ and $B=\la V_{B}, S_{B}, \chi_{B}, W_{B}, \ell_{B} \ra$ be two (generalized) simplicial models. A \emph{morphism} $f: \ A \rightarrow B$ of simplicial models consists of a morphism of simplicial complexes from $\la V_A,S_A \ra$ to $\la V_B,S_B\ra$, such that, $f(W_A)\subseteq W_B$, for all $v \in V_A$, $\chi_B(f(v))=\chi_A(v)$, and $\ell_B(f(v))=\ell_A(v)$. 
\end{definition}

The protocol that we use to solve a task will be specified by a communication pattern model $P$, as defined in \cref{sec:dynamics}.
Then, the protocol simplicial model will be defined as the updated model $\cP = \cI \odot P$.
Since both the protocol model $\cP$ and the task complex $\cT$ are defined as products, they come with first projection morphisms $\pi_\cI : \cP \to \cI$ and $\pi_\cI : \cT \to \cI$.
The role of these morphisms is to recall, for a given final state or output value, from which input state it originally came from.
With this data, we can reformulate the solvability of a task as follows:


\begin{definition}
\label{thm:Kripketasksolv2}
 A task ${\mathcal{T}}$ is \emph{solvable} using the protocol $P$ 
  if there exists a morphism $\delta: \cP\rightarrow \cT$ such that
$\pi_{\cI}\, \circ\, \delta=\pi_\cI$, i.e., the diagram of simplicial complexes below commutes.
\end{definition}

\begin{center}
\begin{tikzpicture}
  \node (s) {$\cP$};
  \node (xy) [below=2 of s] {$\mathcal{T}$};
  \node (x) [left=of xy] {$\cI$};
  \draw[<-] (x) to node [sloped, above] {$\pi_\cI$} (s);
  \draw[->, right] (s) to node {$\delta$} (xy);
  \draw[->] (xy) to node [below] {$\pi_\cI$} (x);
\end{tikzpicture}
\end{center}

The intuition behind this definition is the following.
A world $X$  in $\cP$ corresponds to a global state that is reachable from input $\pi_\cI(X)$ in $\cI$. 
The morphism~$\delta$ takes~$X$ to a world $\delta(X)=(w_{\cI},w_{\cO})$ of $\cT$. 
The commutativity of the diagram expresses the fact that both $X$ and $\delta(X)$ correspond to the same input
assignment $u$.
Now consider a single vertex $v \in X$ with $\chi(v) = a \in A$. Then, agent $a$
decides its value solely according to its knowledge in $\cP$: if another
world $X'$ contains $v$, then $\delta(v) \in \delta(X) \cap \delta(X')$, meaning
that $a$ has to decide the same value in both situations.

\subsection{Knowledge gain as a logical tool for task solvability}

\label{sec:knowledge-gain}

In~\cite{gandalf-journal}, to prove that map $\delta$ of \cref{thm:Kripketasksolv2} does not exist, we rely on a key property of our logic called ``knowledge gain''.
This principle says that agents cannot acquire new knowledge along morphisms of simplicial models. Namely, what is known in the image of a morphism was already known in the domain.
Thus, to prove that the simplicial map~$\delta : \cP \to \cT$ cannot exist, we have to find a formula~$\phi$ such that:
\begin{enumerate}
\item \label{item:obstruction-task} that $\phi$ is true in every world of~$\cT$,
\item \label{item:obstruction-protocol} and that $\phi$ is false in at least one world of~$\cP$.
\end{enumerate}
Then by the knowledge gain property, the map $\delta$ does not exist. Such a formula $\phi$ is called a \emph{logical obstruction}.
Intuitively, the formula $\phi$ describes some amount of knowledge which is a necessary condition to be able to solve the task $\cT$ (\cref{item:obstruction-task}), and is not achieved using protocol~$\cP$ (\cref{item:obstruction-protocol}).
%

\subparagraph{Knowledge gain for guarded formulas.}

In~\cite{gandalf-journal}, the formulas $\phi$ that could be used as obstruction formulas were all positive formulas.
Here, in the presence of process crashes, we need an additional restriction: $\phi$ must be a \emph{guarded formula}, which we define now.
Formally, the fragment of \emph{guarded positive epistemic formulas} $\phi \in \PL{K,\text{alive}}$ is defined by the grammar:
\begin{align*}
	\varphi ::=\; & \aliveprop{B} \Rightarrow \psi_B \mid \varphi \land \varphi \mid \varphi \lor \varphi \mid 
	D_U\varphi \mid C_U \varphi & & U, B \subseteq A,\; \psi_B \in \PropL{B}\\
	\psi_B ::=\; & p \mid \neg \psi_B \mid \psi_B \land \psi_B & & p \in \AP_B	
\end{align*}
where:
\begin{itemize}
\item $C_U$ is the common knowledge operator \cite{FaginHV92}, which is the least solution of the equation $C_U \varphi = \varphi \vee \bigvee_{u \in U} K_u(C_U\varphi)$, and whose semantics on a simplicial model $M=\la V_{M}, S_{M}, \chi_{M}, W_{M}, \ell_{M} \ra$ is given as follows: 
for every $Y$ in $W_M$ reachable from $X$ in $W_M$ following a sequence of simplexes sharing a $U$-colored simplex, $M,Y \models \varphi$
\item  the formula $\aliveprop{B}$ stands for $\bigwedge_{a \in B} \aliveprop{a}$, 
\item and the formula $\psi_B \in \PropL{B}$ is a \emph{propositional formula} restricted to the agents in $B$. It can only contain atomic propositions concerning the agents in $B$, and no modal operator.
\end{itemize}

\begin{theorem}[knowledge gain, revisited]
\label{thm:lose-knowledge-3}
Consider simplicial models $\C=\la V_{\C}, S_{\C}, \chi_{\C}, W_{\C}, \ell_{\C} \ra$
and
$\D=\la V_{\D}, S_{\D}, \chi_{\D}, W_{\D}, \ell_{\D} \ra$, and a  morphism $f : \C \to \D$.
Let $\varphi \in \PL{K,\textup{alive}}$ be a guarded positive epistemic formula.
Then $\D,f(X) \models \varphi$ implies $\C,X \models \varphi$.
\end{theorem}
\begin{proof}
We proceed by induction on the structure of the guarded positive formulas $\varphi$.
The inductive cases are obvious to prove for $\land$, $\neg$ and $\lor$. 

$D_B$ and $C_U$ are easily proven as follows. 
Suppose $\D,Y \models D_B \varphi$, then
$\D,w' \models \phi \text{ for all } w' \in W_\D \text{ such that } B \subseteq \chi(w' \inter Y)$. As $\phi$ is a guarded positive formula, by induction, $\C, X \models \varphi$, for all $X'$ such that $f(X')\in W_\D$ and $B \subseteq \chi(f(X')\cap Y)$. But $f$ is a morphism of pointed simplicial models so $X'$ being in $W_\C$ implies $f(X')\in W_\D$ and $\chi(f(X')\cap Y)=\chi(f(X')\cap f(X))=\chi(f(X'\cap X))=\chi(X'\cap X$ hence by induction, $\C, X \models \varphi$, for all $X' \in W_\C$ such that $B \subseteq \chi(X'\cap X)$ meaning that $\C, X \models D_B \varphi$. 

Suppose now that $\D,Y \models C_U \varphi$. 
Then for every $w'$ in $W_\D$ reachable from $Y$ in $W_\D$ following a sequence of simplexes sharing a $U$-colored simplex, $\D,w' \models \varphi$. In particular, for every $X'$ in $W_\C$ reachable from $X$ in $W_\C$ following a sequence of simplexes sharing a $U$-colored vertex, $f(X')$ is reachable from $Y=f(X)$ following a sequence of simplexes sharing a $U$-colored simplex, because $f$ is a morphism of pointed simplicial model. By induction hypothesis, as $\D,w'\models \varphi$, and $\phi$ is a guarded positive formula, $\C, X' \models \varphi$, hence $\C, X\models C_U \varphi$. 

For the base case, assume $\phi = \aliveprop{B} \Rightarrow \psi_B$ for some set of agents $B \subseteq A$ and 
some propositional formula $\psi_B \in \PropL{B}$.
We distinguish two cases.
Either some agent $a \in B$ is dead in the world $X$, in which case $\C,X \models \phi$ is true.
Or all agents in $B$ are alive in $X$, and since $f(X) \subseteq Y$ (because $f$ is a morphism of pointed simplicial models), all agents in $B$ are also alive in $Y$.
Thus, we have $\D,Y \models \psi_B$.
Moreover, since $f$ is a morphism, we know that $\ell_Y(f(v)) = \ell_X(v)$ for all $v$ in $X$.
So all atomic propositions in $\AP_{B}$ have the same truth value in the worlds $X$ and $Y$.
As a consequence $\D,Y \models \psi_B$ implies that $\C,X \models \psi_B$, and thus $\C,X \models \phi$ as required.
\end{proof}

Note that this theorem is slightly different from the one considered in the conference version of this paper~\cite{stacs22}.
First, we consider epistemic formulas with common knowledge, that was not considered as an operator in \cite{stacs22}.
Second, since we are working with a more general definition of simplicial models, the notion of morphism is also slightly different and the proofs is adapted in consequence.

\subsection{Extended example: consensus in synchronous broadcast protocols}
\label{sec:extended-example}

We are now equipped to study the following distributed computing problem: how to prove that the consensus task cannot be solved in the synchronous broadcast model with one crash failure.
Rather than the impossibility result itself, which is well known, our main focus here is to showcase how the epistemic logic machinery developed in this paper can be used to establish an impossibility proof in distributed computing.

\subparagraph*{Input model.}

We will be working with three agents (a.k.a.\ processes) $a$, $b$ and $c$.
For the binary consensus task, each of them starts the computation with an input value, either $0$ or $1$.
The \emph{initial simplicial model} $\cI = \la V_{\cI}, S_{\cI}, \chi_{\cI}, W_{\cI}, \ell_{\cI} \ra$ modeling the initial states of the processes is depicted below. 
Each of the 8 facets of $\cI$ represents
a possible initial configuration for the agents $a$, $b$ and $c$, with possible input values 0 or 1.
We denote by $\inputprop{a}{v}$ the atomic proposition meaning that ``agent~$a$ has input value~$v$''.
On the figure below, the labeling of a vertex is indicated by a subscript: $\ell_{\cI}(a_i)=\{\inputprop{a}{i}\}$, $\ell_{\cI}(b_i)=\{\inputprop{b}{i}\}$ and $\ell_{\cI}(c_i)=\{\inputprop{c}{i}\}$, for $i=0, 1$. The set of worlds $W$ associated to the simplicial model $\cI$ is composed of the 8 facets modeling the initial states when all processes are alive, plus the 12 edges in thick black below, modeling the possible states of any pair of processes, the third one being dead:

\begin{center}
	\begin{tikzpicture}[x  = {(1cm,0.5cm)},
		scale=1,
		cloudgrey/.style={draw=black,thick,circle,fill=cyan,inner sep=1pt,minimum size=11pt},cloud/.style={draw=black,thick,circle,fill=magenta!60,inner sep=1pt,minimum size=11pt}, cloudblack/.style={draw=black,thick,circle,fill=green,inner sep=1pt,minimum size=11pt}
		]] 
		
		\coordinate (a0) at (0,0,1);
		\coordinate (a1) at (2,2,1);
		\coordinate (b0) at (0,2,1);
		\coordinate (b1) at (2,0,1);
		\coordinate (c0) at (1,1,0);
		\coordinate (c1) at (1,1,2);
		
		
		\draw [thick, line width=0.5mm, draw=black,fill opacity=0.15,fill=gray] (a0) -- (b0) -- (c0) -- cycle;
		\draw [thick, line width=0.5mm, draw=black,fill opacity=0.15,fill=gray] (a0) -- (b0) -- (c1) -- cycle;
		\draw [thick, line width=0.5mm, draw=black,fill opacity=0.15,fill=gray] (a0) -- (b1) -- (c0) -- cycle;
		\draw [thick, line width=0.5mm, draw=black,fill opacity=0.15,fill=gray] (a0) -- (b1) -- (c1) -- cycle;
		\draw [thick, line width=0.5mm, draw=black,fill opacity=0.15,fill=gray] (a1) -- (b0) -- (c0) -- cycle;
		\draw [thick, line width=0.5mm, draw=black,fill opacity=0.15,fill=gray] (a1) -- (b0) -- (c1) -- cycle;
		\draw [thick, line width=0.5mm, draw=black,fill opacity=0.15,fill=gray] (a1) -- (b1) -- (c0) -- cycle;
		\draw [thick, line width=0.5mm, draw=black,fill opacity=0.15,fill=gray] (a1) -- (b1) -- (c1) -- cycle;
		
		\node[cloudgrey] at (0,0,1) {$a_0$};
		\node[cloudblack, fill=gray!50!green] at (1,1,0) {$c_0$};
		\node[cloud] at (0,2,1) {$b_0$};
		\node[cloudgrey] at (2,2,1) {$a_1$};
		\node[cloudblack] at (1,1,2) {$c_1$};
		\node[cloud] at (2,0,1) {$b_1$};
	\end{tikzpicture}
\end{center}

\subparagraph*{The synchronous broadcast model with crash failures.}

\label{sec:example-DC-informal}

Let us first explain informally the synchronous broadcast model with one crash failure.
In order to keep the pictures small and $2$-dimensional, we assume here that there is a single input simplex, where agents $a,b,c$ always start the computation with input values $v_1,v_2,v_3$, respectively; but we keep in mind that in general, an agent does not know in advance the inputs of the others.
(\cref{fig-synchEvol} depicts a less degenerate situation where we start with the full binary input complex.)

Thus, at the beginning of the computation, the local state of each process is its input value. Then, communication occurs via synchronized rounds. At each round:
\begin{itemize}
\item Each process sends its own local state to all other processes, in an unspecified order.
\item At most one process may crash per round. When a process crashes, it simply stops sending messages.
Under the \emph{detectable crashes} assumption, a process may not crash after successfully sending all of its messages. This ensures that at least one of the other processes is able to witness the crash.
\item The round ends when all the non-faulty processes have finished sending their messages. Each non-faulty process then updates its local state by appending all the messages that it received during the round; we then proceed to the next round.
\end{itemize}
Due to the synchronous nature of this model, whenever a round ends and some process~$P$ has not received a message from process~$Q$, process~$P$ immediately knows that~$Q$ has crashed.

In the following, we focus on modeling a single round of computation.
The resulting simplicial model is the one depicted in \cref{ex:message-passing}. It can be computed using the communication pattern $P_{\text{detectable}}$ of \cref{ex:broadcast-comm-pattern}.
Now we can make explicit the labeling of vertices.
Thus, here all~$a$-labeled vertices~$u_a$ (in blue in \cref{ex:message-passing}) have labeling $\ell(u_a)=\inputprop{a}{v_1}$, all $b$-labeled vertices $u_b$ (in red) have $\ell(u_b)=\inputprop{b}{v_2}$ and all $c$-labeled vertices $u_c$ (in green) have $\ell(u_c)=\inputprop{c}{v_3}$.

\subparagraph*{Output model.}

The output model $\cO= \la V_{\cO}, S_{\cO}, \chi_{\cO}, W_{\cO}, \ell_{\cO} \ra$ is depicted below: 

\begin{center}
\begin{tikzpicture}[scale=1,
    cloudgrey/.style={draw=black,thick,circle,fill=cyan,inner sep=1pt,minimum size=11pt},cloud/.style={draw=black,thick,circle,fill=magenta!60,inner sep=1pt,minimum size=11pt}, cloudblack/.style={draw=black,thick,circle,fill=green,inner sep=1pt,minimum size=11pt}
    ] 
        
\coordinate (a0) at (0,0);
\coordinate (a1) at (4,0);
\coordinate (b0) at (0,2);
\coordinate (b1) at (4,2);
\coordinate (c0) at (1,1);
\coordinate (c1) at (5,1);

\draw [thick, line width=0.5mm, draw=black,fill opacity=0.15,fill=gray] (a0) -- (b0) -- (c0) -- cycle;
\draw [thick, draw=black,fill opacity=0.15,line width=0.5mm, fill=gray] (a1) -- (b1) -- (c1) -- cycle;

\node[cloudgrey] at (0,0) {$a^0$};
\node[cloudblack] at (1,1) {$c^0$};
\node[cloud] at (0,2) {$b^0$};
\node[cloudgrey] at (4,0) {$a^1$};
\node[cloudblack] at (5,1) {$c^1$};
\node[cloud] at (4,2) {$b^1$};
\end{tikzpicture}
\end{center}

In this model, there are 8 worlds: there are the two facets, modeling the fact that the three agents are still alive, and they either all decide 0, or all decide 1. There are also the 6 edges in thick black modeling the fact that two among three agents are still alive when the protocol completes, deciding either 0 or 1. The decision values are indicated as a superscript on agent's names. As a simplicial model, we declare the labelling on vertices empty. 

\subparagraph*{Binary consensus task specification.}


The task specification simplicial model is given as a relation between input and output, hence, has as underlying simplicial complex a subcomplex of the product complex $\cI\times \cO$.
The product $\cI\times \cO$ is depicted below.
Its worlds are a subset of the set-theoretic product $W_{\cI}\times W_{\cO}$, which is composed of exactly two copies of $W_I$, with worlds being all edges, shown below as thick lines, and triangles, in grey:  

\begin{center}
\begin{tikzpicture}[x  = {(1cm,0.5cm)},
    scale=1,
    cloudgrey/.style={draw=black,thick,circle,fill=cyan,inner sep=1pt,minimum size=11pt},cloud/.style={draw=black,thick,circle,fill=magenta!60,inner sep=1pt,minimum size=11pt}, cloudblack/.style={draw=black,thick,circle,fill=green,inner sep=1pt,minimum size=11pt}
    ]]
        
\coordinate (a0) at (0,0,1);
\coordinate (a1) at (2,2,1);
\coordinate (b0) at (0,2,1);
\coordinate (b1) at (2,0,1);
\coordinate (c0) at (1,1,0);
\coordinate (c1) at (1,1,2);


\draw [thick, line width=0.5mm, draw=black,fill opacity=0.35,fill=gray] (a0) -- (b0) -- (c0) -- cycle;
\draw [thick, line width=0.5mm, draw=black,fill opacity=0.35,fill=gray] (a0) -- (b0) -- (c1) -- cycle;
\draw [thick, line width=0.5mm, draw=black,fill opacity=0.35,fill=gray] (a0) -- (b1) -- (c0) -- cycle;
\draw [thick, line width=0.5mm, draw=black,fill opacity=0.35,fill=gray] (a0) -- (b1) -- (c1) -- cycle;
\draw [thick, line width=0.5mm, draw=black,fill opacity=0.35,fill=gray] (a1) -- (b0) -- (c0) -- cycle;
\draw [thick, line width=0.5mm, draw=black,fill opacity=0.35,fill=gray] (a1) -- (b0) -- (c1) -- cycle;
\draw [thick, line width=0.5mm, draw=black,fill opacity=0.35,fill=gray] (a1) -- (b1) -- (c0) -- cycle;
\draw [thick, line width=0.5mm, draw=black,fill opacity=0.35,fill=gray] (a1) -- (b1) -- (c1) -- cycle;

\node[cloudgrey] at (0,0,1) {$a_0^0$};
\node[cloudblack, fill=gray!70!green] at (1,1,0) {$c_0^0$};
\node[cloud] at (0,2,1) {$b_0^0$};
\node[cloudgrey] at (2,2,1) {$a_1^0$};
\node[cloudblack] at (1,1,2) {$c_1^0$};
\node[cloud] at (2,0,1) {$b_1^0$};
        
\coordinate (a0) at (3,-3,-2); 
\coordinate (a1) at (5,-1,-2);
\coordinate (b0) at (3,-1,-2);
\coordinate (b1) at (5,-3,-2);
\coordinate (c0) at (4,-2,-3);
\coordinate (c1) at (4,-2,-1);


\draw [thick, line width=0.5mm, draw=black,fill opacity=0.35,fill=gray] (a0) -- (b0) -- (c0) -- cycle;
\draw [thick, line width=0.5mm, draw=black,fill opacity=0.35,fill=gray] (a0) -- (b0) -- (c1) -- cycle;
\draw [thick, line width=0.5mm, draw=black,fill opacity=0.35,fill=gray] (a0) -- (b1) -- (c0) -- cycle;
\draw [thick, line width=0.5mm, draw=black,fill opacity=0.35,fill=gray] (a0) -- (b1) -- (c1) -- cycle;
\draw [thick, line width=0.5mm, draw=black,fill opacity=0.35,fill=gray] (a1) -- (b0) -- (c0) -- cycle;
\draw [thick, line width=0.5mm, draw=black,fill opacity=0.35,fill=gray] (a1) -- (b0) -- (c1) -- cycle;
\draw [thick, line width=0.5mm, draw=black,fill opacity=0.35,fill=gray] (a1) -- (b1) -- (c0) -- cycle;
\draw [thick, line width=0.5mm, draw=black,fill opacity=0.35,fill=gray] (a1) -- (b1) -- (c1) -- cycle;

\node[cloudgrey] at (3,-3,-2) {$a_0^1$};
\node[cloudblack, fill=gray!70!green] at (4,-2,-3) {$c_0^1$};
\node[cloud] at (3,-1,-2) {$b_0^1$};
\node[cloudgrey] at (5,-1,-2) {$a_1^1$};
\node[cloudblack] at (4,-2,-1) {$c_1^1$};
\node[cloud] at (5,-3,-2) {$b_1^1$};
\end{tikzpicture}
\end{center}
In the picture above, the left binary sphere represents the possible output situations where the processes decide 0.
The binary sphere on the right represents situations where processes decide 1.
In a vertex, the subscript represents the input value of a process, and the superscript represents the output.
The labeling is taken from the input value, e.g., $\ell_{\cI}(a^j_i)=\{\inputprop{a}{i}\}$.

In the presence of crashes, there are various ways to specify the consensus task~\cite{ksetconsensus}.
The first one is called the validity axiom (SV1) in \cite{ksetconsensus}:
\begin{quote}
``The decision of any correct process is equal to the input of some correct process.''
\end{quote}
In that case, we should take out among the simplices of the corresponding task specification $\cT\subseteq W_I\times W_O$ the triangle with all inputs at 1, for the left copy of the binary sphere (which corresponds to deciding 0), the triangle with all inputs 0 for the right copy of the binary sphere (which corresponds to deciding 1), and also take out the 3 edges with all 1s on the left sphere, and the 3 edges with all 0s on the right sphere, as worlds, leading to the following picture, simplicial model $\cT_1$: 

\begin{center}
\begin{tikzpicture}[x  = {(1cm,0.5cm)},
    scale=1,
    cloudgrey/.style={draw=black,thick,circle,fill=cyan,inner sep=1pt,minimum size=11pt},cloud/.style={draw=black,thick,circle,fill=magenta!60,inner sep=1pt,minimum size=11pt}, cloudblack/.style={draw=black,thick,circle,fill=green,inner sep=1pt,minimum size=11pt}
    ]] 
        
\coordinate (a0) at (0,0,1);
\coordinate (a1) at (2,2,1);
\coordinate (b0) at (0,2,1);
\coordinate (b1) at (2,0,1);
\coordinate (c0) at (1,1,0);
\coordinate (c1) at (1,1,2);

\draw [thick, line width=0.5mm, draw=black] (a0) -- (b0) -- (c0) -- cycle;
\draw [thick, line width=0.5mm, draw=black] (a0) -- (b0) -- (c1) -- cycle;
\draw [thick, line width=0.5mm, draw=black] (a0) -- (b1) -- (c0) -- cycle;
\draw [thick, line width=0.5mm, draw=black] (a1) -- (b0) -- (c0) -- cycle;
\draw [thick, line width=0.5mm, draw=black] (a1) -- (b0) -- (c1);
\draw [thick, line width=0.5mm, draw=black] (a1) -- (c0) -- (b1);
\draw [fill opacity=0.35,fill=black] (a0) -- (b0) -- (c0) -- cycle;
\draw [fill opacity=0.35,fill=gray] (a0) -- (b0) -- (c1) -- cycle;
\draw [fill opacity=0.35,fill=gray] (a0) -- (b1) -- (c0) -- cycle;
\draw [fill opacity=0.35,fill=gray] (a0) -- (b1) -- (c1) -- cycle;
\draw [fill opacity=0.35,fill=gray] (a1) -- (b0) -- (c0) -- cycle;
\draw [fill opacity=0.35,fill=gray] (a1) -- (b0) -- (c1) -- cycle;
\draw [fill opacity=0.35,fill=gray] (a1) -- (b1) -- (c0) -- cycle;

\node[cloudgrey] at (0,0,1) {$a_0$};
\node[cloudblack, fill=gray!70!green] at (1,1,0) {$c_0$};
\node[cloud] at (0,2,1) {$b_0$};
\node[cloudgrey] at (2,2,1) {$a_1$};
\node[cloudblack] at (1,1,2) {$c_1$};
\node[cloud] at (2,0,1) {$b_1$};

\coordinate (a0) at (3,-3,-2); 
\coordinate (a1) at (5,-1,-2);
\coordinate (b0) at (3,-1,-2);
\coordinate (b1) at (5,-3,-2);
\coordinate (c0) at (4,-2,-3);
\coordinate (c1) at (4,-2,-1);

\draw [thick, line width=0.5mm, draw=black](a0) -- (b1) -- (c0);
\draw [thick, line width=0.5mm, draw=black](a0) -- (b1) -- (c1) -- cycle;
\draw [thick, line width=0.5mm, draw=black](a1) -- (b0) -- (c1) -- cycle;
\draw [thick, line width=0.5mm, draw=black](a1) -- (b1) -- (c0) -- cycle;
\draw [thick, line width=0.5mm, draw=black](a1) -- (b1) -- (c1) -- cycle;

\draw [thick, draw=black,fill opacity=0.35,fill=gray] (a0) -- (b0) -- (c1) -- cycle;
\draw [thick, draw=black,fill opacity=0.35,fill=gray] (a0) -- (b1) -- (c0) -- cycle;
\draw [thick, draw=black,fill opacity=0.35,fill=gray] (a0) -- (b1) -- (c1) -- cycle;
\draw [thick, draw=black,fill opacity=0.35,fill=gray] (a1) -- (b0) -- (c0) -- cycle;
\draw [thick, draw=black,fill opacity=0.35,fill=gray] (a1) -- (b0) -- (c1) -- cycle;
\draw [thick, draw=black,fill opacity=0.35,fill=gray] (a1) -- (b1) -- (c0) -- cycle;
\draw [thick, draw=black,fill opacity=0.35,fill=gray] (a1) -- (b1) -- (c1) -- cycle;

\node[cloudgrey] at (3,-3,-2) {$a_0$};
\node[cloudblack, fill=gray!70!green] at (4,-2,-3) {$c_0$};
\node[cloud] at (3,-1,-2) {$b_0$};
\node[cloudgrey] at (5,-1,-2) {$a_1$};
\node[cloudblack] at (4,-2,-1) {$c_1$};
\node[cloud] at (5,-3,-2) {$b_1$};
\end{tikzpicture}
\end{center}

\noindent
Another, weaker specification of the consensus task is the validity axiom (RV1) of~\cite{ksetconsensus}: 
\begin{quote}
``The decision of any correct process is equal to the input of some process.''
\end{quote}
In that case, we should take out among the simplices of the corresponding task specification $\cT\subseteq W_I\times W_O$ the triangle with all inputs at 1, for the left copy of the binary sphere (which corresponds to deciding 0), the triangle with all inputs 0 for the right copy of the binary sphere (which corresponds to deciding 1), but this time keep the 3 edges with all 1s on the left sphere, and the 3 edges with all 0s on the right sphere, as worlds, leading to the following picture, simplicial model $\cT_2$: 

\begin{center}
\begin{tikzpicture}[x  = {(1cm,0.5cm)},
    scale=1,
    cloudgrey/.style={draw=black,thick,circle,fill=cyan,inner sep=1pt,minimum size=11pt},cloud/.style={draw=black,thick,circle,fill=magenta!60,inner sep=1pt,minimum size=11pt}, cloudblack/.style={draw=black,thick,circle,fill=green,inner sep=1pt,minimum size=11pt}
    ]] 
        
\coordinate (a0) at (0,0,1);
\coordinate (a1) at (2,2,1);
\coordinate (b0) at (0,2,1);
\coordinate (b1) at (2,0,1);
\coordinate (c0) at (1,1,0);
\coordinate (c1) at (1,1,2);

\draw [thick, line width=0.5mm, draw=black] (a0) -- (b0) -- (c0) -- cycle;
\draw [thick, line width=0.5mm, draw=black] (a0) -- (b0) -- (c1) -- cycle;
\draw [thick, line width=0.5mm, draw=black] (a0) -- (b1) -- (c0) -- cycle;
\draw [thick, line width=0.5mm, draw=black] (a1) -- (b0) -- (c0) -- cycle;
\draw [thick, line width=0.5mm, draw=black] (a1) -- (b0) -- (c1);
\draw [thick, line width=0.5mm, draw=black] (a1) -- (c0) -- (b1);
\draw [thick, line width=0.5mm, draw=black] (a1) -- (b1) -- (c1) -- cycle;
\draw [fill opacity=0.35,fill=black] (a0) -- (b0) -- (c0) -- cycle;
\draw [fill opacity=0.35,fill=gray] (a0) -- (b0) -- (c1) -- cycle;
\draw [fill opacity=0.35,fill=gray] (a0) -- (b1) -- (c0) -- cycle;
\draw [fill opacity=0.35,fill=gray] (a0) -- (b1) -- (c1) -- cycle;
\draw [fill opacity=0.35,fill=gray] (a1) -- (b0) -- (c0) -- cycle;
\draw [fill opacity=0.35,fill=gray] (a1) -- (b0) -- (c1) -- cycle;
\draw [fill opacity=0.35,fill=gray] (a1) -- (b1) -- (c0) -- cycle;

\node[cloudgrey] at (0,0,1) {$a_0$};
\node[cloudblack, fill=gray!70!green] at (1,1,0) {$c_0$};
\node[cloud] at (0,2,1) {$b_0$};
\node[cloudgrey] at (2,2,1) {$a_1$};
\node[cloudblack] at (1,1,2) {$c_1$};
\node[cloud] at (2,0,1) {$b_1$};
        
\coordinate (a0) at (3,-3,-2); 
\coordinate (a1) at (5,-1,-2);
\coordinate (b0) at (3,-1,-2);
\coordinate (b1) at (5,-3,-2);
\coordinate (c0) at (4,-2,-3);
\coordinate (c1) at (4,-2,-1);

\draw [thick, line width=0.5mm, draw=black] (a0) -- (b0) -- (c0) -- cycle;
\draw [thick, line width=0.5mm, draw=black](a0) -- (b1) -- (c0);
\draw [thick, line width=0.5mm, draw=black](a0) -- (b1) -- (c1) -- cycle;
\draw [thick, line width=0.5mm, draw=black](a1) -- (b0) -- (c1) -- cycle;
\draw [thick, line width=0.5mm, draw=black](a1) -- (b1) -- (c0) -- cycle;
\draw [thick, line width=0.5mm, draw=black](a1) -- (b1) -- (c1) -- cycle;

\draw [thick, draw=black,fill opacity=0.35,fill=gray] (a0) -- (b0) -- (c1) -- cycle;
\draw [thick, draw=black,fill opacity=0.35,fill=gray] (a0) -- (b1) -- (c0) -- cycle;
\draw [thick, draw=black,fill opacity=0.35,fill=gray] (a0) -- (b1) -- (c1) -- cycle;
\draw [thick, draw=black,fill opacity=0.35,fill=gray] (a1) -- (b0) -- (c0) -- cycle;
\draw [thick, draw=black,fill opacity=0.35,fill=gray] (a1) -- (b0) -- (c1) -- cycle;
\draw [thick, draw=black,fill opacity=0.35,fill=gray] (a1) -- (b1) -- (c0) -- cycle;
\draw [thick, draw=black,fill opacity=0.35,fill=gray] (a1) -- (b1) -- (c1) -- cycle;

\node[cloudgrey] at (3,-3,-2) {$a_0$};
\node[cloudblack, fill=gray!70!green] at (4,-2,-3) {$c_0$};
\node[cloud] at (3,-1,-2) {$b_0$};
\node[cloudgrey] at (5,-1,-2) {$a_1$};
\node[cloudblack] at (4,-2,-1) {$c_1$};
\node[cloud] at (5,-3,-2) {$b_1$};
\end{tikzpicture}
\end{center}





\subparagraph*{Impossibility of (SV1) consensus in one round.}

As well known \cite{Dolev}, consensus cannot be reached in a synchronous architecture with at most $f$ failures in less than $f+1$ rounds. Here we exemplify this result, in logical terms, in the case $f=1$, showing that consensus needs at least 2 rounds to be solvable. 

For the asynchronous wait-free architecture, it is well known that consensus is not solvable (in any number of rounds) see e.g. \cite{FischerL82}, and it is well known that on the epistemic logic side, this comes from the impossibility of reaching common knowledge among agents \cite{Moses2016,gandalf-journal}. In this paper, we propose a new logical obstruction, also based on common knowledge, that works for the case of synchronous architectures. The main idea is that in synchronous architectures, there is a way to tell whether agents have died or are still alive. This is reflected by the knowledge gain theorem, Theorem \ref{thm:lose-knowledge-3}. 

We are now ready to consider the following formulas, for $i=0,1$: 
$$\phi_i=C_A\left(
\mathop{\bigwedge}\limits_{B \subseteq A} \left(\aliveprop{B}\Rightarrow \mathop{\bigvee}\limits_{b \in B} \inputprop{b}{i}\right)
\right)$$
\noindent where $C_A$ is the common knowledge operator for the set of agents $A$. 
The formula $\phi_0\vee \phi_1$ is actually specifying axiom (SV1): indeed in the corresponding task specification model $\cT_1$, $\phi_0$ holds in the left component, whereas $\phi_1$ holds in the right component. 

Now, we check that neither $\mathcal{A},w \models \phi_0$ nor 
$\mathcal{A},w \models \phi_1$.
In $\cT_1$, all simplexes which are worlds are connected to one another. In particular, the triangle which is labelled by $\{\inputprop{a}{0}, \inputprop{b}{0}, \inputprop{c}{0}\}$ is connected to the triangle which is labelled by $\{\inputprop{a}{1}, \inputprop{b}{1}, \inputprop{c}{1}\}$.
Hence, by the semantics of Section \ref{sec:gensimpmod}, agents in $A$ cannot have common knowledge of either $\mathop{\bigwedge}\limits_{B \subseteq A} \left(\aliveprop{B}\Rightarrow \mathop{\bigvee}\limits_{b \in B} \inputprop{b}{0}\right)$ nor $\mathop{\bigwedge}\limits_{B \subseteq A} \left(\aliveprop{B}\Rightarrow \mathop{\bigvee}\limits_{b \in B} \inputprop{b}{1}\right)$.

Let us now consider the following formulas, for $j=0,1$: 
$$\psi_j=C_A\left(
\mathop{\bigwedge}\limits_{B \subseteq A} \left(\aliveprop{B}\Rightarrow \mathop{\bigvee}\limits_{a \in A} \inputprop{a}{j}\right)
\right)$$
Now, the formula $\psi_0\vee \psi_1$ is specifying axiom (RV1): indeed, similarly to the previous case, in the corresponding task specification model $\cT_2$, $\phi_0$ holds in the left component, whereas $\phi_1$ holds in the right component. 

The conclusion holds in a similar manner, binary consensus even with the weak requirement (RV1) cannot be solved in one round in the synchronous broadcast protocol model. 





\section{Conclusion}


In this work, we have extended the simplicial model approach to epistemic logic, so that to account for the case in which some agents may die, and others may know, or not know, they are dead. 

On the model-theoretic side, this implied to decorate simplicial models with subsets of simplexes that are the observable worlds, in the corresponding Kripke model approach. On the logical side, this made us move from $\Sfive$ to $\KBfour$ and other axioms, according to the choices we can make about the knowledge of agents' deaths. 

This paper 
has further ramifications. First, another generalization can be made using (semi-) simplicial sets instead of simplicial complexes, see \cite{goubaultSemisimplicialSetModels2023}, which  deepens the discussion of this paper about distributed knowledge. Second, it is natural to view our simplicial complex models decorated with observable worlds as hypergraphs. This is developed in another sequel~\cite{hypergraph} where we further discuss the ways predicates should be attached to worlds, or to agents (as point of views in~\cite{hypergraph}) or both. 

There are still numerous extensions to this work, to be considered. Indeed, more applications to distributed computing should be developed; in particular, extending the  logical obstruction to the solvability of set agreement by Yagi and Nishimura~\cite{yagiNishimura2020TR} to the synchronous setting where $f$ processes may crash to obtain the lower bound of~\cite{ChaudhuriHLT00,HERLIHY20001} showing that $\lfloor f/k\rfloor +1$ rounds are needed to solve $k$-set agreement.
This  calls for a more in-depth discussion of temporal extensions of our epistemic logics, to account for the evolution of knowledge in distributed computed, through communication, extending the DEL approach which we originally presented in~\cite{gandalf-journal}. 

\bibliography{bibliography.bib}

%
%

%

\end{document}